\documentclass{amsart}
\pdfoutput=1

\usepackage{amssymb}
\usepackage{amsthm}
\usepackage{amsfonts}
\usepackage{amsmath}
\usepackage{pmboxdraw}
\usepackage{verbatim} 
\usepackage{graphicx}
\usepackage{color}
\usepackage[colorlinks=true, citecolor=blue, filecolor=black, linkcolor=black, urlcolor=black]{hyperref}
\usepackage{cite}
\usepackage[normalem]{ulem}
\usepackage{enumerate}
\usepackage{enumitem}
\usepackage{todonotes}
\usepackage{marginnote}

\usepackage{tcolorbox}
\tcbuselibrary{most}

\newcommand\subsubsec[1]{\medskip\noindent\textbullet~\emph{\textbf{#1.}}~}

\def\Phiplus{\Phi^+}
\def\Phiminus{\Phi^-}

\def\bp{{\bar\partial}}

\def\bfs{\boldsymbol}

\def\pa{\partial}
\def\sm{\setminus}

\def\ti{\tilde}
\def\wh{\widehat}
\def\wt{\widetilde}

\def\id{ \mathrm{id}}

\def\Im{ \mathrm{Im}}

\DeclareMathOperator{\Sing}{Sing}

\def\FF{\mathcal{F}}
\def\LL{\mathcal{L}}
\def\OO{\mathcal{O}}

\def\XX{\mathcal{X}}

\def\CC{\mathcal{C}}

\def\C{\mathbb{C}}

\def\E{\mathbf{E}}
\def\H{\mathbb{H}}

\def\R{\mathbb{R}}
\def\S{\mathbb{S}}
\def\A{\mathbb{A}}

\newcommand{\im}{\operatorname{Im}}

\theoremstyle{plain}
\newtheorem*{thm*}{Theorem}
\newtheorem{thm}{Theorem}[section]
\newtheorem{lem}[thm]{Lemma}

\newtheorem{prop}[thm]{Proposition}

\theoremstyle{definition}
\newtheorem*{egs*}{Examples}
\newtheorem*{def*}{Definition}
\newtheorem{eg*}[thm]{Example}

\theoremstyle{remark}
\newtheorem{rem}{Remark}
\newtheorem*{rmk*}{Remark}
\newtheorem*{rmks*}{Remarks}

\numberwithin{equation}{section}

\begin{document}
\title[Conformal field theory in a multiply connected domain]{Conformal field theory of Gaussian free fields \\ in a multiply connected domain}


\author{Tom Alberts}
\address{Department of Mathematics, University of Utah, Salt Lake City, UT 84112, USA \\
School of Mathematics, Korea Institute for Advanced Study, Seoul, 02455, Republic of Korea} 
\email{alberts@math.utah.edu} 

\author{Sung-Soo Byun}
\address{Department of Mathematical Sciences and Research Institute of Mathematics, Seoul National University, Seoul 151-747, Republic of Korea}
\email{sungsoobyun@snu.ac.kr} 

\author{Nam-Gyu Kang} 
\address{School of Mathematics, Korea Institute for Advanced Study, Seoul, 02455, Republic of Korea} 
\email{namgyu@kias.re.kr} 

\thanks{Tom Alberts was supported by NSF grant DMS-1811087, by Simons Foundation grant MPS-TSM-00002584, and by the KIAS Scholar program. Sung-Soo Byun was supported by the POSCO TJ Park Foundation (POSCO Science Fellowship), by the New Faculty Startup Fund at Seoul National University and by the National Research Foundation of Korea funded by the Korea government (NRF-2016K2A9A2A13003815, RS-2023-00301976). Nam-Gyu Kang was supported by by Samsung Science and Technology Foundation (SSTF-BA1401-51) and by the National Research Foundation of Korea (NRF-2019R1A5A1028324), and by a KIAS Individual Grant (MG058103) at Korea Institute for Advanced Study.}


\begin{abstract}
We implement a version of conformal field theory (CFT) that gives a connection to SLE in a multiply connected domain. Our approach is based on the Gaussian free field and applies to CFTs with central charge $c \leq 1$. In this framework we introduce the generalized Eguchi-Ooguri equations and use them to derive the explicit form of Ward's equations, which describe the insertion of a stress tensor in terms of Lie derivatives and differential operators depending on the Teicm\"{u}ller modular parameters. Furthermore, by implementing the BPZ equations, we provide a conformal field theoretic realization of an SLE in a multiply connected domain, which in particular suggests its drift function, and construct a class of martingale observables for this SLE process. 
\end{abstract}

\date{\today}

\maketitle


\newtcolorbox{example}[1]{breakable,
colbacktitle=gray!50!white, fonttitle=\bfseries, coltitle=black, title=Example: {#1}}

\section{Introduction and Main results}
 
In his classical work \cite{Ko1916}, Koebe introduced several classes of the so-called \textit{canonical} multiply connected domains. By their very definition, the canonical domains can be uniquely determined by specifying a finite number of parameters in the moduli space. From the complex geometry viewpoint, these play a role of the Teichm\"{u}ller parameters \cite{Hub06} that describe the deformation of a surface within a particular homeomorphism class. Compared to simply connected domains, one of the principal difficulties arising in a multiply connected domain is that the non-trivial topology of the domain can significantly impact the behavior of the functions defined on it. For instance, the Hadamard variational formula, which describes the infinitesimal changes of the Green's function under specific perturbations of the underlying domain, requires capturing the deformations of \textit{each} boundary component of a multiply connected domain, see e.g. \cite{KMZ05}. 

These same difficulties persist in studying conformal field theory (CFT) in multiply connected domains. In general, one of the main purposes of CFT is to compute correlation functions of a family of conformal fields, such as the operator product expansion (OPE) family generated from the basic fields. These describe the statistical correlations of certain physical observables in a system and can sometimes be expressed in terms of geometric functions of the underlying domain. In addition to computing correlation functions, CFT also aims to derive certain functional equations among them, which provide powerful tools to predict the critical behavior of the statistical physics models. As with the philosophy above, one expects that the description of correlation functions, as well as their functional equations, reveal significant differences when developing the CFT in a more complicated geometry as they may depend on its moduli space in a non-trivial way.

In this paper, we implement a version of CFT with central charge $c \le 1$ to present its link to a forward SLE in a multiply connected domain. The fields that we consider are all constructed from the Gaussian free field (GFF) with Dirichlet boundary condition, a Gaussian random field whose correlations depend on the geometry of the domain through its Green's function. In this theory, there is a well-established way of computing correlation functions based on Wick's formula, operator product expansions, and the Coulomb gas formalism. There is also an infinitesimal calculus for perturbations of CFT correlation functions, which allows one to express differential equations satisfied by the correlation functions without having explicit expressions for the functions themselves. Fundamental examples of this calculus include Ward's equations and the Belavin-Polyakov-Zamolodchikov (BPZ) equations. The former expresses the perturbations of correlation functions to the correlation function of the underlying random field when computed after the insertion of a \textbf{stress tensor}. The latter equates a very specific perturbation of correlation functions to the correlation function of the underlying random field when computed after the insertion of the one-leg operator. The BPZ equations, in particular, provide the precise form of the null-vector equation, a differential equation satisfied by the partition function that encodes important properties of the system. In a simply connected domain, such a theory and its application to the theory of Schramm Loewner evolution (SLE) has been well-developed, see e.g. \cite{BB02,BB03,Ky06,RBGW07,KM13,KM17,KM21, FW03} and references therein. We also refer to \cite{BCDV14,Gru06,Pel19} for reviews.

Let us turn our attention to Ward's equation, which best features the dependence of non-trivial moduli space. 
As in classical differential geometry, the Lie derivative of a conformal field describes the infinitesimal change of the field when it is transported along the flow of a vector field. 
In CFT, this provides a characteristic property of a stress tensor (if it exists). 
It is well known that the Gaussian free field $\Phi$ has a stress tensor of the form 
\begin{equation*}
A=-\frac12 J \odot J, \qquad J=\pa \Phi,
\end{equation*}
where $\odot$ is Wick's product, see e.g. \cite{KM13}. 
In a simply connected domain, Ward's equation states that the insertion of a stress tensor is equivalent to taking a Lie derivative with respect to a proper vector field with a prescribed pole. 
For instance, for a holomorphic differential (a conformally covariant field) $X$ in the OPE family of the GFF with conformal dimension $\lambda$, Ward's equation states that
\begin{equation} \label{Ward upper half plane}
2\,\E\,A(\zeta) X(z)= \Big( k_\zeta(z) \pa_z + \lambda \, k'_\zeta(z) \Big) \E\,X(z), \qquad k_\zeta(z)= \frac{2}{\zeta-z}, 
\end{equation}
where all fields are evaluated in the upper-half plane $\mathbb{H}$. 
Here, the right-hand side of this identity is a Lie derivative of $X$ with respect to the Loewner vector field $k_\zeta$ in $\mathbb{H}$. 
Ward's equation can be generalized to a tensor product of fields in the OPE family of the GFF.

If the underlying domain is not simply connected, the contributions to Ward's equation from each boundary component (except for the boundary component where an SLE grows) must be taken into account.
To provide a clearer illustration let us consider a doubly connected domain as an example. For given modular parameter $r>0$, we take the cylinder $\{ z \in \C : 0 < \im z <r\}/\langle z \mapsto z + 2\pi \rangle$ as a reference domain.
Then in this cylinder uniformization, Ward's equation for a tensor product $X$ of fields in the OPE family of the GFF reads as 
\begin{align}
2\,\E\,A(\zeta) X & =  \LL_{v_\zeta}\E\, X + \frac{1}{\pi}\int_{\gamma}	\E\,A(\xi) X\, d\xi \label{Ward annulus v1}
\\
&= \LL_{v_\zeta}\E\, X + \pa_r \E\,X, \label{Ward annulus v2}
\end{align}
where $\zeta \in [-\pi,\pi],$ $\gamma=[ir-\pi,ir+\pi]$ is the upper boundary component, and $v_\zeta$ is the Loewner vector field in the cylinder. 
Compared to \eqref{Ward upper half plane}, Ward's equation of the form \eqref{Ward annulus v1} clearly shows the contribution from a boundary component. 
On the one hand, the identity \eqref{Ward annulus v2} reveals an interesting (and indeed far from being obvious) relation between the boundary integral under the insertion of a stress tensor and the differentiation with respect to the modular parameter.
In addition, this leads to a differential form of the BPZ equation, which also plays a crucial role in the theory of annulus SLE, cf \cite{BKT23}.
Ward's equation \eqref{Ward annulus v2} was introduced in the pioneering work of Eguchi and Ooguri \cite{EO87} in the complex torus of genus one using a path integral formalism. 
A particular advantage of considering a doubly connected domain or a complex torus is that it enables explicit computations using well-known special functions such as Jacobi theta functions.
With this strategy, \eqref{Ward annulus v2} was rigorously proved in \cite{BKT23,KM17} using Wick's calculus and several functional equations between special functions. 

Despite the extensive study of the Eguchi-Ooguri version of Ward's equation for the genus one case, a corresponding formulation for a multiply connected domain has not, to our knowledge, been proposed in either the physics or mathematics literature. 
This slow progress may be attributed, at least in part, to the challenges arising from the higher dimensional moduli space in such domains and the absence of explicit formulas for the classical special functions. 
In this paper, we address this problem by developing a formalism for the Eguchi-Ooguri version of Ward's equation in a multiply connected domain.
In addition to this, by implementing the BPZ equations, we introduce the SLE drift function from the viewpoint of CFT and construct a large class of the associated SLE martingale observables.

\subsection{Basic setup}
\label{Subsection_setup}

To introduce our results we first recall some basic notions of multiply connected domains and CFT.
Among the canonical multiply connected domains, the \emph{chordal standard domains} are the simplest setting for describing CFT and its applications to SLE. 
By definition, a chordal standard domain is the upper-half plane $\H$ minus a finite number of horizontal slits $[z_k^l, z_k^r]$, with $k=1,\dots, g$ and $\Im\,z_k^l = \Im\,z_k^r > 0$. Often we denote a chordal standard domain by $\H_g$, with the notation emphasizing the number of slits but suppressing their locations. 
By a well-known uniformization theorem, any multiply connected domain $D$ is conformally equivalent to a chordal standard domain. 

Our formulation of CFT crucially relies on the Green's function $G_D$ of a domain $D$, with zero Dirichlet boundary condition. 
Letting $\wt{G}_D(\cdot, z_0)$ be the harmonic conjugate of $G_D(\cdot, z_0)$ we define the complex Green's functions by
\begin{equation*} 
\displaystyle G_{D}^+:=\frac { G_{D}+i\wt{G}_{D}}{2}, \qquad G_D^{-}= \overline{ G_D^+ }. 
\end{equation*} 
Then a multivalued holomorphic function $G_{D}^+(\cdot, z_0)$ is defined up to periods. 
In a general multiply connected domain the Green's function can be expressed in terms of a special transcendental function called the \textit{Schottky-Klein prime function} \cite{Baker95, CM07}. 
In Subsection \ref{Subsection_Green functions} we present another representation of the Green's function using the Riemann theta function and the \textit{Schottky double construction}.

The basic random field that underlies our CFT is the GFF. We denote it by $\Phi$ and consider it as a formal random function $\Phi : D \to \R$ that satisfies with Dirichlet boundary conditions on each boundary component.
Occasionally we add the subscript $(0)$ to $\Phi$, which indicates that we consider the central charge $c=1$ theory. We extend this to $c \leq 1$ later on. 
The $1$- and $2$-point correlation functions of $\Phi$ are given by 
\begin{equation*}
\E\, \Phi(\zeta)=0, \qquad \E\, \Phi(\zeta)\Phi(z)=2G_{D}(\zeta,z).
\end{equation*}
As $\Phi$ is a Gaussian field these formulas characterize its law, which is conformally invariant due to the conformal invariance of the Green's function. The Gaussian law also allows one to apply \textit{Wick's formulas} to polynomial functions of the formal collection of random variables $\{ \Phi(z) : z \in D \}$ and their derivatives, see \cite{Jan97} for an exposition of Wick's formulas.

Correlation functions in CFT are then computed via a combination of Wick's formulas and the \textit{operator product expansion}. OPE is an important binary operation on conformal fields, serving as a tool for analyzing the algebraic structure of CFT. If a field $X$ is holomorphic, then the OPE is defined as a formal Laurent series expansion near the diagonal within correlations:
\begin{equation*}
X(\zeta)Y(z)=: \sum \big(X *_n Y(z)\big) (\zeta -z )^n, \qquad \text{as } \zeta \to z.
\end{equation*}
In particular, the zeroth coefficient, which we simply denote by $X*Y \equiv X\ast_0Y$, is called the OPE product of $X$ and $Y$. 
In other words, the field $X \ast Y$ is obtained by subtracting the divergent terms in the OPE of the fields $X$ and $Y$ and then taking the limit as $\zeta \to z$. This is a standard technique used in field theory to extract finite results from otherwise divergent calculations.
For instance, the OPE of two GFFs is given by
\begin{equation*}
\Phi \ast \Phi(z)= 2u(z,z)+\Phi(z) \odot \Phi(z) , \qquad u(\zeta,z):=G_{D}(\zeta,z)+\log|\zeta-z|. 
\end{equation*}
Here the function $u(z,z)$ is called the domain constant for the Dirichlet boundary condition, see e.g. \cite[Section 4]{BF06}. 
The OPE family of $\Phi$ is the algebra over $\C$ spanned by the generators $1$, all mixed derivatives of $\Phi$, and the so-called \textit{OPE exponentials}, under the operations $(+, \ast)$.

To implement CFT with central charge $c \le 1$, one needs to define the central/background charge modification of $\Phi$. 
For this purpose, given points $q_k$ on a boundary component $C_0$ and a parameter $b \in \R$, we consider a divisor 
\begin{equation*}
\bfs{\beta}= \sum \beta_k \cdot q_k, 
\end{equation*}
with the \emph{neutrality condition}
\begin{equation} \label{NC_b}
\int \bfs{\beta} =  b \chi
\end{equation}
where $\chi=2-2g$ is the Euler characteristic of a compact Riemann surface of genus $g$. 

Then in terms of a conformal map $w: (D, C_0,q )\mapsto (\H_g, \R, \infty)$, which is uniquely determined if we impose hydrodynamic normalization at infinity, we define the background charge modification $\Phi_{ \bfs \beta }$ of $\Phi$ as 
\begin{equation} \label{Phi beta}
\Phi_{ \bfs \beta }(z)=\Phi(z) - 2b \arg w'(z) - 2 \sum \beta_k\, \Im\, G^+_D(w(q_k),w(z)).  
\end{equation}
In particular, if $\bfs \beta_\infty= b \chi \cdot q$, then 
\begin{equation*}
\Phi_{\bfs \beta_\infty }(z) = \Phi(z)-2b \arg w'(z).
\end{equation*}

This form \eqref{Phi beta} of the background charge modification is induced from the Schottky double construction of the GFF and CFT on a compact Riemann surface of genus $g$.
In particular, the neutrality condition comes from a version of the Gauss-Bonnet theorem, see Subsection~\ref{Subsection_background}. 
This modification gives rise to CFT with central charge 
\begin{equation} \label{def of central charge}
c=1-12b^2.
\end{equation}
From $\Phi_{\bfs \beta}$ we also construct the \textit{OPE family} $\mathcal{F}_{ \bfs \beta }$, the algebra of conformal fields (over $\C$) that is generated by the constant field $1$, all mixed derivatives of $\Phi_{\bfs \beta}$, and the OPE exponentials. The addition of the algebra is the regular addition of fields, while the multiplication is the OPE multiplication. The OPE family $\mathcal{F}_{\bfs \beta}$ is the natural family of fields for which we can develop an infinitesimal calculus for the correlation functions and make the connection with SLE theory.

To develop the infinitesimal calculus we now recall the definition of the Lie derivatives and a stress tensor. 
For a conformal Fock space field $X$, the Lie derivative $\LL_v X$ of a smooth vector field $v$ is defined by
\begin{equation*}
( \LL_v X \| \phi ) = \frac{d}{dt} \Big|_{t=0} (X \| \phi \circ \psi_{-t} ),
\end{equation*}
where $\psi_t$ is a local flow of $v$, and $\phi$ is a given local chart.
We further define its $\C$-linear and anti $\C$-linear parts $\LL_v^\pm$ by 
\begin{equation*}
\LL_v^+ :=\frac{\LL_v-i\LL_{iv}}{2},\qquad \LL_v^- :=\frac{\LL_v + i\LL_{iv}}{2}. 
\end{equation*}
The Lie derivative of a conformal Fock space field can be computed from the field's transformation law. In particular, consider a differential $X$ of conformal dimension $[\lambda, \lambda_*]$, which means that its transformation law satisfies the covariance rule
\[
X = (h')^\lambda (\overline{h'})^{\lambda_*} \wt{X} \circ h,
\]
where $X = (X || \phi)$ and $\tilde{X} = (X || \tilde{\phi})$ is the evaluation of $X$ on different charts $\phi, \tilde{\phi}$, and $h$ is the holomorphic transition map between two overlapping charts $\phi, \tilde{\phi}$. The Lie derivatives of a $(\lambda, \lambda_*)$-differential are 
\begin{equation*}
\LL_v^+ X = ( v\pa +\lambda v' ) X, \qquad \LL_v^- X = ( \bar{v}\bp + \lambda_* \overline{v'} ) X.
\end{equation*}
As previously mentioned, the Lie derivatives are closely related to the notion of a stress tensor in CFT. 
By definition, a conformal field $X$ has a stress tensor $A$ if 
\begin{itemize}
 \item $A$ is a holomorphic quadratic differential, and
 \item the residue form of Ward's identity 
\begin{equation*}
\LL_v^+ X(z)=\frac{1}{2\pi i}\oint_{(z)}vA X(z), \qquad \LL_v^- X(z)=-\frac{1}{2\pi i}\oint_{(z)} \bar{v}\bar{A} X(z)
\end{equation*}
holds for all smooth vector fields $v$.
\end{itemize}
It can be shown that all elements of the OPE family $\mathcal{F}_{ \bfs \beta }$ (which is generated from the Gaussian free field $\Phi_{ \bfs{\beta} }$) have the common stress tensor
\begin{equation*}
A_{ \bfs{\beta} }:= -\frac12 J \odot J +\Big(ib\pa- j_{ \bfs{\beta} } \Big)J, \qquad j_{ \bfs{\beta}} = \E\, J_{ \bfs{\beta} }. 
\end{equation*}
Here $J_{ \bfs{\beta} } = \pa \Phi_{ \bfs{\beta} }$ and $J= \pa \Phi.$
Although $A_{ \bfs \beta }$ is not a member of the OPE family $\mathcal{F}_{ \bfs \beta }$, the Virasoro field
\begin{equation*}
T_{ \bfs{\beta} } :=-\frac{1}{2}J_{ \bfs{\beta} } \ast J_{ \bfs{\beta} }+ib\,\pa J_{ \bfs{\beta} } = A_{ \bfs{\beta} } + \E \,T_{ \bfs{\beta} }
\end{equation*}
is in $\mathcal{F}_{\bfs \beta}$. The first equality is by definition while the second is a computation in OPE calculus. It can be shown that the deterministic term $\E\,T_{ \bfs{\beta} }$ is a Schwarzian form of order $c/12$, where $c$ is given by \eqref{def of central charge}.

\subsection{Main results}

We now introduce our results. 
Let $v_{\zeta, \H_g}$ be the Loewner vector field with sole simple pole at $\zeta$ in a chordal standard domain $\H_g$. 
In general, the Loewner vector field in a multiply connected domain can be expressed in terms of the complex excursion reflected Poisson kernel.
Let $\gamma: [0, \infty) \rightarrow \H_g$ be a simple curve with $\gamma(0) \in \R$. 
Then by \cite[Theorem 3.1]{BF08} (see also \cite[Theorem 6.14]{Dren11}), the family of conformal mappings $g_t: \H_g \setminus \gamma_t \to \H_g(t)$, where $\H_g(t)$ is a chordal standard domain, satisfies the \emph{chordal Loewner equation} 
\begin{equation}\label{chL flow}
\displaystyle \pa_t g_t(z)=-v_{\xi_t, \H_g(t)}(g_t(z)), \qquad (g_0(z)=z).
\end{equation}
See Subsection~\ref{Subsec_Loewner vector field} for more details on the Loewner vector field. 

Recall that a chordal standard domain $\H_g$ is characterized by the horizontal slits $[z_k^l, z_k^r]$ with $k=1,\dots, g$.
In terms of the Teichm\"uller parameters $\{ z_k^l, z_k^r\}$, we define the differential operator
\begin{equation} \label{Teich diff op}
\nabla_{ \mathbb{H}_g } := \sum_{k=1}^g \Big( v_\zeta( z_k^l ) \pa_{ z_k^l }+v_{ \zeta }( \bar{z}_k^l ) \bp_{ z_k^l }+ v_\zeta( z_k^r ) \pa_{ z_k^r }+v_{ \zeta }( \bar{z}_k^r ) \bp_{ z_k^r } \Big).
\end{equation}
Here, we shall use the convention $\nabla_{ \mathbb{H}_g } = 0$ if $g=0$.
We further define
\begin{equation}
\nabla_{\mathbb{H}_g, \bfs q } := \nabla_{ \mathbb{H}_g }+ \sum_k v_\zeta(q_k) \partial_{q_k}, 
\end{equation}  
where $\pa_{q_k} =\pa + \bp $ is the differential operator with respect to the real variable $q_k.$ 
Note that while the notation $\nabla_{ \mathbb{H}_g }= \nabla_{ \mathbb{H}_g , \bfs 0} $ treats the underlying domain as a chordal domain $\mathbb{H}_g$ \emph{without} marked points (thus $\nabla_{ \mathbb{H}_g }$ gives the $g \ge 1$ counterpart of the operator $\partial_r$ in \eqref{Ward annulus v2} for the $g=1$ case), the operator $\nabla_{ \mathbb{H}_g , \bfs q }$ emphasizes the underlying geometry as $(\mathbb{H}_g, \bfs{q})$, a domain with marked points. The latter is more natural in the context of CFT with background charge modifications. 
We intentionally introduce the two operators $\nabla_{\mathbb{H}_g}$ and $\nabla_{\mathbb{H}_g, \bfs q}$, as their origins in Ward's equation \eqref{eq: Ward_no_integral} below are different: the first one comes from the generalized Eguchi-Ooguri equation, whereas the second comes from the residue form of Ward's identity at each of the points in $\bfs q$. This distinction also makes the presentation of the proof in Section~\ref{Section_Ward and BPZ} clearer.  

We remark that the operator $\nabla_{ \mathbb{H}_g, \bfs q }$ also naturally appears in the context of the commutation relation of SLE in a multiply connected domain due to Dub\'{e}dat, see \cite[p.1835]{Du07}.

We have the following form of Ward's equation.

\begin{thm}[\textbf{Ward's equations of Eguchi-Ooguri type}] \label{Thm_Ward} 
Fix a background charge $\bfs \beta$ supported on $\R$, and assume it satisfies the neutrality condition \eqref{NC_b}. For any tensor product $X$ of fields in the OPE family $\FF_{ \bfs \beta }$ of the GFF $\Phi_{ \bfs \beta }$, we have 
\begin{align}\label{eq: Ward_w_integral_term}
\begin{split}
2 \, \E A_{ \bfs \beta }(\zeta) X& = \Big( \LL_{v_\zeta}^+ + \LL_{ v_{\bar{\zeta}} }^- \Big) \E X 
\\
&\quad - \frac{1}{2\pi i} \int_{\pa \H_g  \setminus \R  } \Big( v_\zeta(\xi)-v_\zeta( \bar{\xi}) \Big) \E A_{ \bfs \beta }(\xi) X\,d\xi + \, \frac{1}{2\pi i} \sum_k \oint_{ (q_k) } v_\zeta(\xi) \E A_{ \bfs \beta }(\xi) X   
\end{split}
\end{align}
where all fields are evaluated in the identity chart of $\H_g$. Here, the Lie derivatives act on neither $q_k$'s nor the endpoints $\{ z_k^l, z_k^r\}$ of slits.
Furthermore, the generalized Eguchi-Ooguri equation
\begin{equation} \label{EO main statement}
 \frac{1}{2\pi i} \int_{\pa \H_g  \setminus \R  } \Big( v_\zeta(\xi)-v_\zeta( \bar{\xi}) \Big)\E  A_{ \bfs \beta }(\xi) X\,d\xi + \nabla_{ \mathbb{H}_g } \E X=0 
\end{equation}
holds, as does the following residue form of Ward’s
identity at marked points of the background charge
\begin{equation} \label{Ward eq:residue form}
\frac{1}{2\pi i} \sum_k \oint_{ (q_k) } v_\zeta(\xi) \E A_{ \bfs \beta }(\xi) X =  \sum_k v_\zeta(q_k) \partial_{q_k} \E X.    
\end{equation}
Together, \eqref{Ward eq:residue form} and the generalized Eguchi-Ooguri equation \eqref{EO main statement} lead to the Ward equation
\begin{align}\label{eq: Ward_no_integral}
2 \, \E A_{ \bfs \beta }(\zeta) X = \left( \LL_{v_\zeta}^+ + \LL_{ v_{\bar{\zeta}} }^- \right) \E X + \nabla_{  \mathbb{H}_g , \bfs q }\, \E X. 
\end{align}
\end{thm}
 
\begin{rem}
In a multiply connected domain the correlation functions $\E X$ depends on the the nodes of $X$ (for example $z_1$ and $z_2$ if $X = \Phi_{\bfs \beta}(z_1) \Phi_{\bfs \beta}(z_2)$), the points $q_k$ in the support of the background charge $\bfs \beta$, and also implicitly on the modular parameters of the domain. By convention, the Lie derivative terms only act on the nodes of $X$ in formulas like \eqref{eq: Ward_w_integral_term} and \eqref{eq: Ward_no_integral}. 
However, it is worth noting that the differential operator $\nabla_{\mathbb{H}_g, \bfs q}$ can also be interpreted as a Lie derivative operator, if one so desires. This is achieved by regarding the field $\E X$ as a $(0,0)$-differential  at the modular parameters, i.e. at the support points $\{q_k\}$ of $\bfs \beta$ and at the endpoints $\{z_k^l,z_k^r\}$ of the slits (and their complex conjugates). From this viewpoint, the insertion of the stress tensor $A_{\bfs \beta}$ into a field in $\mathcal{F}_{\bfs \beta}$ in a multiply connected domain genuinely acts as a ``total'' Lie derivative, acting not only the nodes of $X$ but also the modular parameters. This viewpoint is consistent with both the physical and geometric perspectives and aligns with the appearance of such terms in \cite{Du09} via Itô calculus. However, the way in which the operator $\nabla_{\mathbb{H}_g, \bfs q}$ arises in our construction of CFT is far from obvious and requires separate treatment, as is evident from \eqref{EO main statement} and \eqref{Ward eq:residue form}.
\end{rem}

\begin{rem} 
For the simplest case when $X=\Phi(z_1)\Phi(z_2)$ and there is no background charge, Ward equations give rise to the differential form of the Hadamard's variational formula
\begin{equation*}
\begin{split}
\frac{d}{dt}G_{D_t}( g_t(z_1), g_t(z_2) ) \Big|_{t=0} = -\Big( \mathcal{L}_{v_p}+\nabla_{ \mathbb{H}_g } \Big) G_D(z_1,z_2)
\end{split}
\end{equation*}
under the Loewner flow $g_t:(D_t,\gamma_t)\to (D,\xi_t),$ where $ p = \gamma_0$. 
We also mention that a similar form of variational formula can be found in the context of the resolvent kernel on a compact polyhedral surface (a compact Riemann surface equipped with a flat metric that has conical singularities), see \cite[Proposition 2]{KK09}.
It is also worth noting that many correlation functions can be expressed in terms of special transcendental functions, such as Schottky-Klein primes or Riemann theta functions, and Ward's equation can provide non-trivial functional equations among these special functions.
\end{rem}

A natural extension of Ward's equations are the so-called \textit{BPZ equations}. There are several distinct ways of introducing the BPZ equations; in our theory they come about as a special limiting form of Ward's equations applied to a particular representation of the correlation functions. For this representation we introduce the \textit{Wick exponential} of a multi-point version of the field $\Phi$. This multi-point field, which we denote by $\Phi\{\bfs \tau\}$, is indexed by double divisors $\bfs \tau= ( \bfs \tau^+, \bfs \tau^- )$ that satisfy the neutrality condition 
\begin{equation*}
\int \bfs\tau^++ \bfs \tau^- =0. 
\end{equation*}
We write $\bfs\tau^{\pm} = \sum \tau_j^{\pm} \cdot z_j$ (noting that some $\tau_j^{\pm}$ may be zero), and define $\Phi\{\bfs \tau\}$ by 
\begin{equation*}
\Phi\{\bfs \tau\} \equiv \Phi \{ \bfs \tau^+, \bfs \tau^- \}= \sum \tau_j^+ \Phi^+(z_j) - \tau_j^- \Phi^-(z_j).
\end{equation*}
Here $\Phi^{\pm}$ are formal 1-point bosonic fields that can be thought of as the holomorphic and anti-holomorphic part of $\Phi$. They are related as 
\begin{equation*}
\Phi=\Phiplus+\Phiminus, \qquad \Phiminus=\overline{\Phiplus}.
\end{equation*} 
We define the Wick's exponential $V^{ \odot }\{\bfs \tau \}$ by
\begin{equation*}
V^{ \odot }\{\bfs \tau \}= e^{ \odot i \Phi \{ \bfs \tau \} }.
\end{equation*}
Then correlation functions of the form $\E\,V^{\odot} \{ \bfs \tau \} X$ are well-defined by Wick calculus. A closely related correlation function is given by $\E\,\OO_{\bfs \beta} \{ \bfs \tau \} X$ where $\OO_{\bfs \beta} \{ \bfs \tau \}$ is the OPE exponential, which we now define.  Given a background charge $\bfs \beta = \sum \beta_k \cdot q_k, (\beta_k \in \R, q_k \in C_0)$ the OPE exponential $\OO_{\bfs \beta}\{\bfs \tau\}$ is
\begin{equation*}
\OO\{ \bfs \tau \} \equiv \OO_{ \bfs \beta } \{\bfs \tau \} := e^{ \ast i \Phi_{ \bfs \beta }\{ \bfs \tau \} }, 
\end{equation*}
where 
\begin{equation*}
\Phi_{ \bfs \beta }\{ \bfs \tau \} := 
\sum \tau_j^+ \Phi^+_{ \bfs \beta }(z_j) - \tau_j^- \Phi^-_{ \bfs \beta }(z_j)
\end{equation*}
and 
\begin{equation*}
\Phi^+_{ \bfs \beta }(z) = \Phi^+(z) + ib \log w'(z) + i \sum \beta_k \,\E\,\Phi^+(q_k) \Phi^+(z) , \qquad \Phiminus_{ \bfs \beta }=\overline{\Phiplus_{ \bfs \beta }},  
\end{equation*}
cf. see \cite[Eq.(3.28) and p.32]{BKT23}. 
Here, a formal correlation $\E\,\Phi^+(\zeta) \Phi^+(z)$ and the OPE exponential are defined in Subsections~\ref{Subsection_chiral} and \ref{Subsection_background}.
For any choice of $\bfs \tau$ the OPE exponential $\OO_{\bfs \beta} \{ \bfs \tau \}$ is in the OPE family $\mathcal{F}_{\bfs \beta}$. In this paper we use it exclusively as an ``insertion" type operator, meaning that we multiply other fields with it before taking the expectation to produce a correlation function. The main effect of this insertion is to insert level two degeneracy into the correlation functions (at least for some specific choices of $\bfs \tau$), which introduces second order differential operators into the limiting form of the Ward equation. In this way we obtain the following form of BPZ equations \cite{BPZ84} in a multiply connected domain.

\begin{thm}[\textbf{BPZ and BPZ-Cardy equations}] \label{Thm_BPZ}
Fix a background charge $\bfs \beta$ supported on a component $C_0$ of the multiply connected domain $\H_g$, and assume $\bfs \beta$ satisfies the neutrality condition \eqref{NC_b}. For $z \in C_0$ that is distinct from the support points of $\bfs \beta$ let 
\begin{equation*}
\Upsilon(z) := \OO_{\bfs\beta} \{ \bfs \alpha \}, \qquad \bfs \alpha= ( a \cdot z - a \cdot \infty, \bfs 0), 
\end{equation*}
where $a \in \R$ satisfies $2a(a+b)=1$ and $b \in \R$ is determined by the neutrality condition \eqref{NC_b} of $\bfs \beta$. 
Let $r_{0} \equiv r_{D,0}$ and $r_{1} \equiv r_{D,1}$ be the first two coefficients in the holomorphic part of the Loewner vector field's Laurent expansion, i.e.
\begin{equation*}
v_\zeta(z) = \frac{2}{\zeta-z}+ r_{0}(z) + r_{1}(z)(\zeta-z) + o(|\zeta-z|), \qquad \zeta \to z. 
\end{equation*}
Then for any tensor product $X$ of fields in the OPE family $\FF_{ \bfs \beta }$, the following holds.
\begin{itemize}
 \item[(i)] \textup{\textbf{(BPZ equation)}} We have
\begin{equation}
\begin{split} \label{BPZ eq}
\frac{1}{a^2} \pa_z^2 \E \Upsilon(z) X & = \Big( \check{\LL}_{v_z}^+ + \LL_{v_{\bar{z}}}^- + \nabla_{ \mathbb{H}_g, \bfs q }  \Big) \E \Upsilon(z) X + \Big[ h_{1,2}\Big(r_{0}'(z)-r_{1}(z)\Big) + 2\,\E\,T_{\bfs\beta}(z) + r_{0}(z) \pa_z \Big] \E \Upsilon(z) X, 
\end{split}
\end{equation}
where all fields are evaluated in the identity chart of $\H_g$.
Here, $h_{1,2}=a^2/2-ab$ and the check indicates that the Lie derivative operator $\check{\LL}_{v_z}^+$ applies to the punctured domain $\H_g\setminus\{z\}$. 
\smallskip 
\item[(ii)] \textup{\textbf{(BPZ-Cardy equation)}} Let $\xi\in\R$ and 
\begin{equation*}
\E_{\Upsilon} X := \frac{ \E \Upsilon(\xi) X }{ \E \Upsilon(\xi) } = \E V^{ \odot } \{ \bfs \alpha \} X. 
\end{equation*}
Then we have 
\begin{equation} \label{BPZ-Cardy}
\frac{1}{a^2} \pa_\xi^2 \E_\Upsilon X + \Big( \frac{2}{a^2} \pa_\xi \log \E \Upsilon(\xi) - r_{0}(\xi) \Big) \pa_\xi \E_\Upsilon X = \Big( \check{\LL}_{v_\xi}^+ + \LL_{v_\xi}^- +  \nabla_{ \mathbb{H}_g , \bfs q }   \Big) \E_\Upsilon X,
\end{equation}
where $\pa_\xi = \pa + \bp$ is the differential operator with respect to the real variable $\xi.$
\end{itemize}
\end{thm}

\begin{rem}[Null-vector equation]
By letting $X=1$ in the BPZ equations \eqref{BPZ eq}, we obtain the null-vector equation satisfied by $\E \Upsilon(z)$: 
\begin{equation}
\begin{split} \label{null vector eq}
\frac{1}{a^2} \pa_z^2 \E \Upsilon(z) & = \nabla_{ \mathbb{H}_g, \bfs q }\, \E \Upsilon(z)  + \Big[ h_{1,2} \Big(r_{0}'(z)-r_{1}(z)\Big) + 2\,\E\,T_{\bfs\beta}(z) + r_{0}(z) \pa_z \Big] \E \Upsilon(z) . 
\end{split}
\end{equation}
\end{rem}

\begin{rem}[Comparison with the BPZ equations in a simply connected domain]
In a simply connected domain, the associated Loewner vector field is given by 
\begin{equation*}
k_{\zeta}(z)= \frac{2}{\zeta-z}
\end{equation*}
in the identity chart of $\H.$
Then the BPZ equation reads as 
\begin{equation}
\begin{split} \label{BPZ eq simply}
\frac{1}{a^2} \pa_z^2 \E \Upsilon(z) X & = \Big( \check{\LL}_{k_z}^+ + \LL_{k_{\bar{z}}}^- +  \sum_k v_\zeta(q_k) \partial_{q_k} \E X    + 2\,\E\,T_{\bfs\beta}(z) \Big) \E \Upsilon(z) X.
\end{split}
\end{equation}
See \cite[Proposition 5.13]{KM13} and \cite[Theorems 7.11 and 7.12]{KM21} for more about the BPZ equations in a simply connected domain.
Compared to \eqref{BPZ eq simply}, the BPZ equations \eqref{BPZ eq} in a multiply connected domain contain the geometric terms 
\begin{equation*}
\Big[ \nabla_{ \mathbb{H}_g }+ h_{1,2}\Big(r_{0}'(z)-r_{1}(z)\Big) + r_{0}(z) \pa_z \Big] \E \Upsilon(z) X.
\end{equation*}
Among these, the Teichm\"uller term $\nabla_{ \mathbb{H}_g }$ naturally comes from the Eguchi-Ooguri type of Ward's equation (Theorem~\ref{Thm_Ward}) and captures the evolution of slits. 
On the other hand, the remaining terms originate from the infinitesimal changes of the Lie derivatives in the chordal standard domain with respect to those in the upper-half plane; more precisely
\begin{equation*}
\left[ h_{1,2}\Big(r_{0}'(z)-r_{1}(z)\Big) + r_{0}(z) \pa_z \right] \E \Upsilon(z) X = \lim_{\zeta \to z} \left( \LL_{v_\zeta}^+- \LL_{k_\zeta}^+ \right) \E \Upsilon(z) X.
\end{equation*}
See Subsection~\ref{Subsection_BPZ} for further details.
\end{rem}

\begin{rem}
The function $r_0$ describes the difference of each side of a small vertical slit attached to the real axis as it evolves under the Loewner flow. 
To be more precise, if we define 
\begin{equation*}
\widehat{v}_\zeta(z) := v_\zeta(z) -r_0(z) 
\end{equation*}
so that 
\begin{equation*}
\widehat{v}_\zeta(z) = \frac{2}{\zeta-z} +O(\zeta-z), \qquad \zeta \to z, 
\end{equation*}
then the image of each side of a vertical slit under the flow generated by $\widehat{v}_\zeta$ has the same leading-order length on $\R$, see \cite{BF06,BF08}.
See also Remark~\ref{Remark_SLE drift} below for its appearance in the SLE drift functions. 
\end{rem}

The Schramm-Loewner evolution $\mathrm{SLE}(\kappa,\Lambda)$ with positive parameter $\kappa$ and the drift function $\Lambda\equiv\Lambda_D$ in a multiply connected domain $D$ is a conformally invariant law on random curves in $D$ described by the Loewner equation \eqref{chL flow} and driving process
\begin{equation} \label{driving SDE}
d \xi_t= \sqrt{\kappa} \, dB_t+ \Lambda_{D_t}(\xi_t) \,dt. 
\end{equation}
Throughout the paper, we consider the drift function associated with $\E \Upsilon (\xi):$ 
\begin{equation} \label{drift function}
\Lambda(\xi):=\kappa \, \pa_\xi \log \E \Upsilon (\xi) - r_{0}(\xi).
\end{equation}
Note that $\Lambda$ implicitly depends on the domain $D$, the background charge $\bfs \beta$, and the marked points in $\bfs \alpha = (a \cdot \xi + \bfs \alpha_0, \bfs 0)$ that are used to define $\Upsilon(\xi)$. The value of $a$ and $\kappa$ is determined by the neutrality condition of $\bfs \beta$. More precisely, if $\bfs \beta$ satisfies neutrality condition \eqref{NC b}, then $\kappa$ and $a$ are determined by the bijections 
\begin{equation*}
b = \sqrt{\kappa/8} - \sqrt{2/\kappa}, \quad a=\sqrt{2/\kappa}.  
\end{equation*}
This gives the central charge
\begin{equation*}
c= 1-12b^2= \frac{ (3\kappa-8)(\kappa-6) }{ 2\kappa }.
\end{equation*}

The relationship between SLE and our CFT is by the following notion of a \textit{martingale observable}. By definition, a non-random field $M$ is a martingale-observable for the SLE system \eqref{chL flow} if the process $M_t(z_1,\ldots,z_n) = (M_{(D_t,\gamma_t,\bfs q)}\,\|\,\id)(z_1,\ldots,z_n)$ is a local martingale on the SLE probability space, at least until the first time that any marked point is swallowed. Note that $M_t=(M \,\|\, g_t^{-1} )$ and thus $M_t$ becomes random by its evaluation on the random charts $g_t$, which in turn inherits its randomness from the stochastic evolution equation of the driving process $\xi_t$.
Using the BPZ-Cardy equations \eqref{BPZ-Cardy}, the next result shows that all fields in $\mathcal{F}_{\bfs \beta}$ lead to martingale observables for a certain class of stochastic evolution equations for $\xi_t$.

\begin{thm}[\textbf{SLE martingale observables}] \label{Thm_SLE MO}
Let the background charge $\bfs \beta$ and the OPE exponential $\Upsilon(z)$ be defined as in Theorem \ref{Thm_BPZ}.
Suppose that $\XX$ is a string of fields in $\FF_{ \bfs \beta }$. 
Then the non-random field 
\begin{equation*}
M = \E_\Upsilon \XX 
\end{equation*}
is a martingale observable for SLE($\kappa,\Lambda_D$) with drift function $\Lambda_D$ given by \eqref{drift function}.
\end{thm}

\begin{rem}[SLE drift function] \label{Remark_SLE drift}
In general, there are many candidates for SLE drift functions. 
Our drift function \eqref{drift function} is motivated by the BPZ-Cardy equation \eqref{BPZ-Cardy}, which leads to the martingale observable property. 
Our drift function is also natural from the viewpoint of the GFF/SLE coupling, see also \cite{IK13,Du09}.
On the other hand, Dub\'{e}dat proposed SLE drift functions using the commutation relation \cite[Section 9]{Du07}.
More precisely, in terms of a partition function $Z$ satisfying a null-vector equation similar to \eqref{null vector eq}, the drift function of the form 
\begin{equation*}
\kappa \, \pa_\xi \log Z(\xi) -r_0(\xi)
\end{equation*}
was considered, see \cite[pp.1834--1835]{Du07}.
Let us also mention that Zhan \cite{Zh04} used the drift function of the form 
\begin{equation*}
\Big(3- \frac{\kappa}{2}\Big) \frac{ \pa_\xi \pa_\eta G_D(\xi,\eta) }{ \pa_\xi G_D(\xi,\eta) }
\end{equation*}
to define the harmonic random Loewner chain in a multiply connected domain.
This in particular gives the scaling limit of the loop-erased random walk \cite{Zh08} for $\kappa=2$.
\end{rem}

\section{Green's functions and Loewner vector field}

In this section we present a concise exposition of the fundamental concepts pertaining to Green's functions and Loewner vector fields in a multiply connected domain.

\subsection{Special functions} \label{Subsection_special functions}

We first collect basic facts on fundamental geometric functions in a multiply connected domain. To make the presentation more concrete we frequently use either chordal standard domains or circular domains as a model of a multiply connected domain. The circular domains are more convenient for expressing the geometric functions such as the Schottky-Klein prime function, while the chordal standard domains are more useful for describing the Loewner evolutions. There is no loss of generality in this approach since every multiply connected domain is conformally equivalent to a chordal standard domain.

\subsubsection{Chordal Standard Domains and their Schottky Doubles}
Our presentation also uses the notion of the \textit{Schottky double} of the multiply connected domain. If $D$ is a multiply connected domain then, informally, its Schottky double is obtained by gluing a copy of $D$ to itself along the boundary components. This makes the Schottky double a compact Riemann surface with genus equal to the number of ``holes'' in the original domain $D$. In many ways the Schottky double is the natural setting for complex analysis, and analytical results on the original domain are a shadow of results on the double.

\subsubsection{Circular Domains and their Schottky Doubles}
We call a domain $D$ circular if it is a finitely connected domain all of whose boundary components are circles. This domain is often employed as a prototypical example of a multiply connected domain in geometric function theory. We assume that the outer boundary $C_0$ is the unit circle and write $C_j$ ($j=1,\cdots,g$) for each boundary component. Let $\delta_j \in \C$ and $r_j >0$ be the centres and radii of the $C_j$'s, i.e.
\begin{equation}\label{circles C_j}
C_j=\{ \zeta \in \C : |\zeta-\delta_j|=r_j \}. 
\end{equation}
By assumption the $\delta_j$ and $r_j$ are such that $C_j \subset \{ z : |z| < 1 \}$ for each $j$, and the $C_j$, $j=1,\ldots,g$ are all disjoint. We sometimes write $\bfs{r} = (r_1, \ldots, r_g)$, $\bfs{\delta} = (\delta_1, \ldots, \delta_g)$, and $\CC(\bfs{r}, \bfs{\delta})$ for the domain
\[
\CC(\bfs{r}, \bfs{\delta}) = \left \{ z \in \C : |z| < 1, |z - \delta_j| > r_j \textrm{ for } j=1,\ldots, g \right \}.
\]

For a circular domain $\CC(\bfs{r}, \bfs{\delta})$ the Schottky double has a straightforward representation derived from reflection across the unit circle. For each $j=1,\ldots,g$ let $C'_j$ be the circle obtained by the reflection of $C_j$ with respect to $C_0$, i.e. $C'_j=\{ \zeta^* \in \C : \zeta \in C_j \}$, where $\zeta^*=1/\bar{\zeta}$. The $2g$ circles $\{ C_j, C_j' : j=1,\dots, g \}$ are known as the \textit{Schottky circles}. The complement of the (union of) the Schottky circles is a model of the Schottky double for the circular domain. This complement is naturally divided into two halves: the part in the interior of the open unit circle and then the reflection of this interior set through the unit circle. In the usual nomenclature these two parts are the ``front'' and ``back'' sides of the Schottky double. The double has no boundary since each point $\zeta \in C_j$ is identified with its reflection $\zeta^* \in C_j'$, and points on the unit circle $\{ |z| = 1 \}$ are interior points of the double.

\subsubsection{Schottky-Klein prime function} 
The Schottky-Klein prime function has many uses, but for our purposes it is a convenient tool for representing meromorphic functions on a compact Riemann surface. We introduce it in the context of circular domains. For each $j=0,1,\cdots, g$, and with $\delta_0 = 0$ and $r_0 = 1$, we define M\"{o}bius maps 
\begin{equation}\label{conj map}
\displaystyle \phi_j(\zeta) := \bar{\delta}_j+\frac{r_j^2}{\zeta-\delta_j}.
\end{equation}
Note that $\phi_j$ conjugates points on the circle $C_j$, i.e. for $\zeta \in C_j$, $\bar{\zeta}=\phi_j(\zeta)$.
Set 
\begin{equation}\label{gen SK}
\displaystyle \theta_j(\zeta):=\bar{\phi}_j(\zeta^{-1})=\delta_j+\frac{r_j^2\zeta}{1-\bar{\delta}_j\zeta}. 
\end{equation}
Then it is easy to see that the image of $C'_j$ under $\theta_j$ is $C_j$. The \emph{Schottky group} $\Theta$ is defined to be the infinite free group of mappings generated by compositions of basic M\"{o}bius maps $\theta_j$ and their inverses $\theta_j^{-1}$, $j=0,1,\cdots, g$. We refer to \cite{MSW02} for more about Schottky group.

From the Schottky group we construct the \emph{Schottky-Klein (S-K) prime function} $\omega(\zeta,z)$. The first step in the construction is to choose a semigroup $\Theta''$ within the Schottky group $\Theta$. Thus $\Theta''$ should be closed under composition but exclude the identity and all inverse maps. For instance, if $\theta_1$, $\theta_1 \theta_2^{-1} \in \Theta''$, then $\theta_1^{-1}$, $\theta_1^{-1}\theta_2 \not \in \Theta''$. Given $\Theta''$ the Schottky-Klein prime function is defined by the infinite product
\begin{equation}\label{S-K ftn}
\displaystyle \omega(\zeta,z):=(\zeta-z) \, \wt{\omega}(\zeta,z), \qquad \wt{\omega}(\zeta,z):= \prod_{\ti{\theta}\in \Theta''} \frac{(\ti{\theta}(\zeta)-z) (\ti{\theta}(z)-\zeta ) }{ (\ti{\theta}(\zeta)-\zeta) (\ti{\theta}(z)-z ) }.
\end{equation}
The function $\omega$ is well defined even though $\Theta''$ is not determined uniquely. See \cite[Sections 4 and 5]{CM05} and \cite{He72} for more about S-K prime function. 
The S-K prime function satisfies the rules 
\begin{equation}\label{SK prime transformation}
 \omega(\zeta,z)=-\omega(z,\zeta), \qquad \overline{\omega(1/\bar \zeta,1/\bar z)}=-\frac{1}{\zeta z} \omega(\zeta,z).
\end{equation}
The skew-symmetry of $\omega$ follows from \eqref{S-K ftn} and the symmetry of $\wt{\omega}$. For the second rule, we refer to \cite[Appendix~A]{CM05}.

\begin{example}{S-K prime function on an annulus}
On the annulus $\A_r = \{ z \in \C : e^{-r} < |z| < 1\}$ the Schottky group has $\theta_1(\zeta)=e^{-2r}\zeta$ as a generator, and we may take 
\[
\Theta'' = \left \{ \theta_1^{(k)} : k \in \mathbb N \right \}
\]
as our semigroup.
Therefore, the S-K prime function $\omega$ in $\A_r$ is represented as 
\begin{equation}\label{S-K annuli}
\displaystyle \omega(\zeta,z)= (\zeta-z)\prod_{k=1}^{\infty} \frac{ (e^{-2rk}\zeta-z) (e^{-2rk}z-\zeta) } { (e^{-2rk}\zeta-\zeta) (e^{-2rk}z-z) }.
\end{equation}
\end{example}

\subsubsection{Riemann theta function}

For a general compact Riemann surface $M$ of genus $g$, let us fix a canonical basis $\{a_j,b_j\}$ of the homology $H_1=H_1(M)$ satisfying $a_j \cdot b_k = \delta_{jk}$ (recall that $a \cdot b$ is the intersection number of the curves $a$ and $b$). On a circular domain the $a_j$ can be chosen as the circles $C_j$, $j=1,\ldots,g$, each of which is identified with the circles $C_j'$. Each $b_j$ can be taken as a curve in the complement of the Schottky circles that connects a $\zeta \in C_j$ to $\zeta^* \in C_j'$. The space of holomorphic $1$-differentials on $M$ is $g$-dimensional, and there is a basis $\{\eta_j \}$ of this space which is uniquely determined by the equations 
\begin{equation} \label{vj def}
\oint_{a_k} \eta_j = \delta_{jk}.
\end{equation}
Given the basis $\{ \eta_j \}$ the period matrix $\tau$ is defined as 
\begin{equation*}
\tau=\{ \tau_{jk}\}, \qquad \tau_{jk}=\oint_{b_k} \eta_j.
\end{equation*}
It is well known that the period matrix is symmetric and $\im \tau$ is positive definite, see e.g. \cite[Chapter III.3]{FK92}. From the period matrix we construct the lattice $\Lambda=\mathbb{Z}^g+ \tau\mathbb{Z}^g$ in $\C^g$ and denote by 
$$
\mathbb{T}_\Lambda \equiv \mathbb{T}_\Lambda^g:= \C^g /\Lambda
$$ 
the Jacobi variety of $M$.
For each $p_0 \in M$, we write
\begin{equation*}
\mathcal{A}_{p_0} \equiv \mathcal{A} : M \rightarrow \mathbb{T}_\Lambda, \qquad p \mapsto \int_{p_0}^{p} (\eta_1,\cdots,\eta_g)^T 
\end{equation*}
for the \emph{Abel-Jacobi} map.

The \textit{Riemann theta function} $\Theta(\cdot| \tau) $ associated with $\tau$ is given as 
\begin{equation*}
\Theta \equiv \Theta(Z|\tau):=\sum_{N \in \mathbb{Z}^g} e^{2\pi i (Z \cdot N +\frac{1}{2}\tau N \cdot N) }, \qquad Z\in \C^g.
\end{equation*}
The theta function is an even, entire function on $\C^g$ and has the periodicity properties
\begin{equation}\label{perio R-theta}
\Theta(Z+N)=\Theta(Z), \qquad \Theta(Z+\tau N)=e^{-2\pi i (Z \cdot N + \frac{1}{2}\tau N \cdot N)} \Theta (Z), \qquad \textrm{ for } N \in \mathbb{Z}^g.
\end{equation}

\subsection{Green's function with Dirichlet/excursion reflected boundary condition} \label{Subsection_Green functions}
We now introduce the Green's function of a multiply connected domain. In contrast to the last section we first phrase the discussion on general domains. Throughout we assume that $D$ is a bounded planar domain whose boundary is $g+1$ disjoint, smooth Jordan curves $C_j$, $j=0,1 \cdots, g$. We take $C_0$ to be the outermost boundary. A prototypical example is given by a circular domain.

For a specified boundary condition and each $z \in D$, the Green's function $\zeta \mapsto G_D(\zeta, z)$ is the unique function such that 
$G_D(\zeta,z) + \log |\zeta-z|$ is harmonic with respect to $\zeta$ throughout the region $D$, including at $z$, and that obeys the prescribed boundary condition. 

We shall consider two fundamental boundary conditions: the zero Dirichlet boundary condition and the excursion reflected (ER) boundary condition. 
The Green's function with ER boundary condition, also referred to as the hydrodynamic or modified Green's function \cite{Sc50,CM05}, is often used in the theory of conformal mappings in a multiply connected domain \cite{FK92}. 
To distinguish Green's functions with Dirichlet and ER boundary conditions, we use the notations $G_D^{Diri}$ and $G_D^{ER}$, respectively. 
The boundary conditions are given as follows. 
\begin{itemize}
 \item \emph{Dirichlet boundary condition}: 
 	\begin{equation} \label{DG cond}
	G_D^{Diri}(\zeta,z)=0 \qquad \text{if} \quad \zeta \in C_j, \quad (j=0,1, \ldots, g).
	\end{equation}
	\item \emph{Excursion reflected boundary condition}: 
	\begin{align}\label{ERG cond}
	\begin{split}
	G_D^{ER}(\zeta,z)&=\begin{cases}
	0 &\text{if}\quad \zeta \in C_0,
	\smallskip	
	\\
	\gamma_j(z) &\text{if} \quad \zeta \in C_j, \quad (j=1, \ldots, g),
	\end{cases}
	\end{split}
	\end{align}
	where $\gamma_j$'s are functions depending only on $z$, and
	\begin{equation*}
	\displaystyle \oint_{C_j} \frac{\pa G_D^{ER}(\zeta,z) }{\pa n_z}\, ds_z = 0, \quad \text{for}\quad j=1, \ldots, g.
	\end{equation*}
	Here, $\pa/\pa n_z$ is the derivative in the direction of the outwarding point normal and $ds_z$ denotes the arclength element with integration variable $z$. These requirements uniquely determine the functions $\gamma_j$.
\end{itemize}

The Green's function with Dirichlet boundary conditions measures local times for Brownian motion on $D$ that is killed when it hits one of the $C_j$. For the ER Green's function the associated stochastic process, called ER Brownian motion, was studied in \cite{CFR16,Dren11}.

\subsubsection{Relation between Dirichlet and ER Green's function}

We denote by $h_j, j=0,1,\ldots, g$, the harmonic functions on $D$ with boundary conditions
\begin{equation}\label{harmonic msrs}
h_j(\zeta)=
\begin{cases}
1 \quad \text{on} & C_j, 
\\
0 \quad \text{on} & C_k, \quad k\neq j.
\end{cases}
\end{equation}
Note that $h_j(\zeta)$ corresponds to the probability that a Brownian motion starting at $\zeta$ exits $D$ through $C_j$.  
It is well known that they can be expressed in terms of $G_D^{Diri}$ as
\begin{equation*}
\displaystyle h_j(\zeta)=-\frac{1}{2\pi} \oint_{C_j} \frac{\pa {G_D^{Diri}}(\zeta,z)}{\pa n_z}\, ds_z.
\end{equation*}
From the $h_j$ we also extract the collection of integrals
\begin{equation}\label{Pkj}
\displaystyle P_{kj}=\frac{1}{2\pi} \oint_{C_k}\frac{\pa h_j(z)}{\pa n_z}ds_z, 
\end{equation}
which are known as the periods of $h_j$. Then it is known that the \textit{period matrix} 
$\displaystyle \bfs{P}=[P_{kj}]_{k,j=1}^g$
is real, symmetric and positive-definite. 
We write $\bfs{h}=(h_1,\ldots,h_g)^T$. The following relation follows from \cite[Proposition 5.2]{Dren11}, \cite[Eq.(45)]{CM07} and \cite[Eq.(6)]{BF06}.

\begin{prop}\label{prop:ER_Diri_relation}
In a multiply connected domain $D$, we have
\begin{align}
 G_D^{Diri}(\zeta,z) & = G_D^{ER}(\zeta,z)-\sum_{j=1}^{g} \gamma_j (z) \, h_j(\zeta)
= G_D^{ER}(\zeta,z)- \Big\langle \bfs{h}(\zeta), \bfs{P}^{-1} \bfs{h}(z) \Big\rangle, \label{ER-Diri 2}
\end{align}
where $\langle \cdot, \cdot \rangle$ is the standard inner product. 
\end{prop}

\subsubsection{Green's functions on circular domains}

In \cite{CM07a}, the authors express the ER Green's function defined in a circular domain in terms of the S-K prime function. 

\begin{prop}
On a circular domain $D$, we have 
\begin{equation}\label{ER Green v1}
\displaystyle G_D^{ER}(\zeta,z)= \frac{1}{2}\log \left| \frac{ \omega(\zeta, \bar{z}^{-1}) \, \omega(\zeta^{-1},\bar{z}) }{\omega(\zeta,z) \, \omega(\zeta^{-1},z^{-1}) } \right| = \log \left| \frac {z \, \omega(\zeta, \bar{z}^{-1})} {\omega(\zeta,z)} \right|.
\end{equation}
and so, by Proposition \ref{prop:ER_Diri_relation},
 \begin{equation}\label{Diri Green v1}
 \displaystyle G_D^{Diri}(\zeta,z)= \log \Big| \frac{z \, \omega(\zeta, \bar{z}^{-1})}{\omega(\zeta,z)} \Big|-\Big\langle \bfs{h}(\zeta), \bfs{P}^{-1} \bfs{h}(z) \Big\rangle. 
 \end{equation}
\end{prop}

The second equality of \eqref{ER Green v1} is from the transformation rule \eqref{SK prime transformation} of the S-K prime function. 

\begin{example}{Green's functions on an annulus}
Let us return to the annulus $\A_r = \{ z \in \C : e^{-r} < |z| < 1\}$ as a reference domain. From formula \eqref{S-K annuli} for the S-K prime function and Proposition 2.2 we have
\begin{equation}\label{ER Green annul}
G_{\A_r}^{ER}(\zeta,z)=-\log \left| \frac{z-\zeta }{1-\zeta \bar{z}} \prod_{k=1}^{\infty} \frac{(1-e^{-2rk}\zeta z^{-1}) (1-e^{-2rk} z \zeta^{-1}) }{ (1-e^{-2rk} \zeta \bar{z}) (1-e^{-2rk} \zeta^{-1} \bar{z}^{-1}) } \right|.
\end{equation}
It is also easy to see that in $\A_r$, 
\begin{equation*}
h_1(\zeta)=-\frac{\log|\zeta|}{r}, \qquad \bfs{P}=P_{11}=\frac{1}{r}.
\end{equation*}
Therefore by \eqref{ER-Diri 2}, the Dirichlet Green's function on $\A_r$ is given by
\begin{equation} \label{Diri Green annul}
 G_{\A_r}^{Diri}(\zeta,z)=
 -\log \left| \frac{z-\zeta }{1-\zeta \bar{z}} \prod_{k=1}^{\infty} \frac{(1-e^{-2rk}\zeta z^{-1}) (1-e^{-2rk} z \zeta^{-1}) }{ (1-e^{-2rk} \zeta \bar{z}) (1-e^{-2rk} \zeta^{-1} \bar{z}^{-1}) } \right|-\frac{\log|\zeta| \log|z|}{r}. 
\end{equation}
\end{example}

\subsubsection{Domain constant and the conformal radius}
Returning to the general setting, as $\zeta \to z$ the Green's functions have the asymptotic expansions 
\begin{align}
G_D^{Diri}(\zeta,z)=\log \frac{1}{|\zeta -z|}+d_D(z) +o(1), \label{domain const}
\\
G_D^{ER}(\zeta,z)=\log \frac{1}{|\zeta -z|}+c_D(z) +o(1). \label{conformal radius}
\end{align}
We recall that the function $d_D$ is called the domain constant, whereas $c_D$ is called the (logarithmic) conformal radius, see e.g. \cite[Section 4]{BF06}.

\subsubsection{Bipolar Green's function on a compact Riemann surface} 
We will also need to make use of Green's functions on the Schottky double, but for this it is more useful to consider the more general \textit{bipolar} Green's function on a compact Riemann surface. We now recall a representation of the bipolar Green's function in terms of the Riemann theta function; the full details can be found in \cite{KM17}. 

Let $M$ be a compact Riemann surface of genus $g$, and let $p,q$ be distinct marked points of $M$. The bipolar Green's function $z \mapsto G_{p,q}(z)$ is determined (up to additive constants) by the requirement that it be harmonic and have the following asymptotic expansion near marked points: in some/any chart
\begin{align*}
G_{p,q}(z)&= \phantom{+}\log\frac{1}{|z-p|}+O(1) \quad (z \rightarrow p),
\\
G_{p,q}(z)&= -\log\frac{1}{|z-q|}+O(1) \quad (z \rightarrow q).
\end{align*}
It follows from \cite[Appendix A]{KM17} that the bipolar Green's function $G_{p,q}$ is expressed as 
\begin{equation}
\label{BG HG}
G_{p,q}(z)=\log \left| \frac{\Theta(\mathcal A(z)- \mathcal A(q)-e)}{\Theta(\mathcal A(z)- \mathcal A(p)-e)} \right|
-2\pi \Big\langle (\im \tau)^{-1} \im[\mathcal A(p)- \mathcal A (q)] , \im \mathcal{A}(z) \Big\rangle,
\end{equation}
where the point $e$ is chosen to be an element of $\mathbb{T}_{\Lambda}$ satisfying $\Theta(e) = 0$ and such that neither of the maps 
\begin{equation*} 
z \mapsto \begin{cases}
 \Theta(\mathcal A(z)-\mathcal A(p)-e)
 \\
 \Theta(\mathcal A(z)-\mathcal A(q)-e)
\end{cases} 
\end{equation*}
is identically zero. 
Note that the expression \eqref{BG HG} depends on the choice of $e$, but not the choice of the base point $p_0$ of the Abel-Jacobi map $\mathcal{A}$. 

\begin{example}{The genus one case} 
We consider the complex torus $\mathbb{T}_{\Lambda}:=\C/\Lambda$ of genus one, where $\Lambda=\mathbb{Z}+\tau \mathbb{Z}$ is the group generated by $z \mapsto z+1, z \mapsto z+\tau$. 
Here, the modular parameter $\tau$ is in the upper-half plane $\mathbb H$. The holomorphic differentials on $\mathbb{T}_{\Lambda}$ form a one-dimensional vector space spanned by $dz$, therefore the Abel-Jacobi map is simply the identity, and according to \eqref{BG HG} the bipolar Green's function $G_{p,q}$ on $M=\mathbb{T}_{\Lambda}$ is given by (up to an additive constant)
\begin{equation}\label{BH g=1}
G_{p,q}(z)=\log \left| \frac{\theta(z-q)}{\theta(z-p)} \right| - 2\pi \, \frac{\im (p-q) \, \im z}{\im \tau}, \qquad \theta(z)\equiv \theta_1(z|\tau).
\end{equation}
\end{example}

For an example of a Riemann surface with boundary, consider the case of the Dirichlet and ER Green's function on a cylinder.

\begin{example}{On a cylinder}
In the cylinder
\begin{equation}\label{eq: cylinder_representation}
\mathcal C_r:=\S_r / \langle z \mapsto z+1 \rangle, \qquad \S_r := \Big \{z \in \C: 0 <\im z < \frac{r}{2\pi} \Big\},
\end{equation}
the Green's functions are expressed in terms of Jacobi theta function $\theta(r,z) \equiv \theta_1(z|\frac{ir}{\pi})$ as
\begin{equation}
G^{ER}_{\mathcal{C}_r}(\zeta,z)= \log \left| \frac{\theta (r,\zeta-\bar{z})}{\theta(r,\zeta-z)} \right|, \qquad 	G_{\mathcal{C}_r}^{Diri}(\zeta,z)=\log \left| \frac{\theta(r,\zeta-\bar{z})}{\theta(r,\zeta-z)} \right| -4\pi^2 \, \frac{\im \zeta \, \im z}{r}	 \label{Diri Green cyl}.
\end{equation}
\end{example}

\subsection{Loewner vector field in a multiply connected domain}\label{Subsec_Loewner vector field}

Let $D$ be a multiply connected domain bounded by $g+1$ smooth Jordan curves $C_j$ ($j=0,1 \cdots, g$). 
We define the \emph{complex ER Green's function} by
\begin{equation} \label{C.ERG}
\displaystyle G^{ER+}_{D} := \frac{1}{2} \left( G^{ER}_{D}+i\wt{G}^{ER}_{D} \right),
\end{equation}
where $\wt{G}^{ER}_D$ is the harmonic conjugate of the ER Green's function. 
We denote by
\begin{equation}\label{ER Po}
H_D^{ER}(\zeta,q):=\frac12 \frac{\pa G_D^{ER}(\zeta,z)}{\pa n_z} \Big|_{z=q}, \qquad (q \in C_0)
\end{equation}
the \emph{ER Poisson kernel}. 
Let us also define \emph{complex ER Poisson kernel} by 
\begin{equation} \label{C.ER Po}
\displaystyle H^{ER+}_{D} = \frac12 \left( H^{ER}_{D} + i \wt{H}^{ER}_{D} \right),
\end{equation}
where $\wt{H}^{ER}_D$ is the harmonic conjugate of the ER Poisson kernel. We mention that the complex ER Poisson kernel is uniquely defined, see e.g. \cite[Proposition 6.4]{Dren11}.

The complex ER Poisson kernel can be used to construct a conformal map from $D$ onto a chordal standard domain. More precisely, the conformal map is given by
\begin{equation}\label{conformal maps chordal}
f(\zeta):= -i a \, H_D^{ER+}(\zeta,q) +b, \qquad (a,b \in \mathbb R_{\ge 0}), 
\end{equation}
which maps $(C_0;q) \to (\R, \infty)$, see e.g. \cite{CM06,Dren11}. 
In particular, if $D$ is contained in the upper-half plane and $q=\infty$, the \emph{canonical mapping}
\begin{equation}\label{cano map}
g(\zeta):= -2i H_D^{ER+}(\zeta,\infty) 
\end{equation}
satisfies the hydrodynamic normalization
\begin{equation}\label{hd nor}
\lim_{\zeta \rightarrow \infty} ( g(\zeta)-\zeta ) =0. 
\end{equation}
See also \cite[Proposition 6.4]{Dren11} for a probabilistic representation. 

The Loewner vector field $v_\zeta(z)$ is the analytic function in $z$ with
\begin{equation*}
\im v_\zeta(z)=- \frac{\pa G_D^{ER}(z,\zeta)}{\pa n_\zeta}, \qquad (\zeta \in C_0), 
\end{equation*}
see e.g. \cite[Theorem 3.1]{BF08}. Then in terms of the ER Poisson kernel, we have 
\begin{equation}\label{def: v_zeta}
v_{\zeta}(z) \equiv v_{\zeta,D}(z) := 2i H^{ER+}_D(z,\zeta). 
\end{equation}

Let $\gamma: [0, \infty) \rightarrow D$ be a simple curve with $\gamma(0) \in \R$ and $g_t$ be the canonical mapping from $D \setminus \gamma[0,t]$ onto $D$. 
Write 
\begin{equation*}
\Omega_t= g_t (D \sm \gamma[0,t] ), \qquad \xi_t= g_t (\gamma(t) ).
\end{equation*}
Then by \cite[Theorem 3.1]{BF08} (see also \cite[Theorem 6.14]{Dren11}), the family of conformal mappings $g_t$ satisfies the \emph{chordal Loewner equation} \eqref{chL flow}, which we recall is
\begin{equation*}
\displaystyle \pa_t g_t(z)=-v_{\xi_t, \Omega_t}(g_t(z)), \qquad g_0(z) = z.
\end{equation*}
Equation \eqref{chL flow} was derived in \cite{Ko50} (resp., \cite{BF06}) for the annulus (resp., disc) with circular slits. See also \cite{CFM23}.

\begin{example}{Loewner vector field on a simply connected domain}
In the upper-half plane $\H$ with Dirichlet boundary conditions we have
\begin{equation} \label{Green's function in H}
 G_{\H}(\zeta,z)= \log \left| \frac{\zeta-\bar{z}}{\zeta-z} \right|, \qquad G_{\H}^{+}(\zeta,z)= \frac{1}{2}\log \left( \frac{\zeta-\bar{z}}{\zeta-z} \right) . 
\end{equation}
On the other hand, by \eqref{Green's function in H}, we have
\begin{equation*}
H_\H(z,x)= \left. \frac12 \frac{d}{dy} \left[ \log \left( \frac{z-x+iy}{z-x-iy} \right) \left( \frac{\bar{z}-x-iy}{\bar{z}-x+iy} \right) \right] \right|_{y=0} = \frac{i}{z-x}- \frac{i}{\overline{z}-x}
\end{equation*}
and 
\begin{equation*}
H_{\H}^+(z, x) = \frac{1}{2} \frac{d}{d y} \log \Big(\frac{z-x+iy}{z-x-iy}\Big) \Big|_{y=0} = \frac{i}{z-x}. 
\end{equation*}
Therefore in this case, the equation \eqref{chL flow} gives rise to the well-known chordal Loewner equation
\begin{equation*}
\pa_t g_t(z)= \frac{2}{g_t(z)-\xi_t}, \qquad g_0(z)=z.
\end{equation*}
\end{example}

\begin{example}{Loewner vector field on an annulus}
Note that by \eqref{Diri Green annul}, we have
\begin{equation*}
G^{ER,+}_{\mathcal{C}_r}(\zeta,z)=\frac12 \log \left( \frac{\theta (\zeta-\bar{z})}{\theta(\zeta-z)} \right).
\end{equation*}
Then by using 
\begin{equation*}
\frac{\theta'(z)}{ \theta(z) } = \zeta(r,z) + \frac13 \frac{ \theta'''(0) }{ \theta'(0) } z,
\end{equation*}
where $\zeta(r,z)$ is the Weierstrass zeta function, we have
\begin{equation*}
H^{ER,+}_{\mathcal{C}_r}(z,x)= \left. \frac12 \frac{d}{dy} \log \left( \frac{\theta (z-x+iy)}{\theta(z-x-iy)} \right) \right|_{y=0} = i \left( \zeta(z-x)+\frac13 \frac{ \theta'''(0) }{ \theta'(0) } (z-x) \right).
\end{equation*}
This gives that in the cylinder $\CC_r$, 
\begin{equation*}
v_p(z)= -2\,\zeta(z-p)- \frac{2}{3} \frac{ \theta'''(0) }{ \theta'(0) } (z-p) .
\end{equation*}
By \cite[Eq.(23.9.3)]{NIST}, $\zeta(z)$ has the asymptotic expansion 
\begin{equation*}
\zeta(z)= \frac{1}{z} +O(z^3) \qquad \textrm{as } z \to 0, 
\end{equation*}
hence it follows that as $p \to z$, 
\begin{equation*}
v_p(z) = \frac{2}{p-z} + \frac23 \frac{\theta'''(0)}{ \theta'(0) } (p-z) +O( |p-z|^2 ).
\end{equation*}
\end{example}

\subsubsection{Asymptotics of $v_{\zeta}$ in a chordal standard domain}
In a general chordal standard domain with $\zeta \in \R$ the Loewner vector field has the asymptotic expansion 
\begin{equation} \label{v diagonal}
v_\zeta(z) = \frac{2}{\zeta-z}+ r_{D,0}(z) + r_{D,1}(z)(\zeta-z) + o(|\zeta-z|), \qquad z \to \zeta,
\end{equation}
see e.g. \cite[Proposition 6.4]{Dren11} and \cite[p.1834]{Du07}. 
In particular, 
\begin{equation} \label{v' diagonal}
v_\zeta'(z) = \frac{2}{(\zeta-z)^2}+ \Big(r_{D,0}'(z)-r_{D,1}(z)\Big) + O(|\zeta-z|), \qquad z \to \zeta. 
\end{equation} 
The $z$-dependence of the coefficients in \eqref{v diagonal} comes from the fact that the (ER) Green's function in the chordal standard domain is not translation invariant, i.e. $G_D^{ER}(\zeta,z) \not = G_D^{ER} (\zeta+t,z+t). $

\section{Gaussian free fields in a multiply connected domain} \label{Section_GFF}

\subsection{GFFs and their correlation functions}\label{sec:GFF_correlation}
Let $D$ be a multiply connected domain bounded by $g+1$ smooth Jordan curves $C_j$, $j=0,1 \cdots, g$. We assume $C_0$ is the outermost boundary component. 
 
Let $\Phi$ be the mean zero GFF with zero Dirichlet boundary condition and 2-point correlation function 
\begin{equation} \label{2pt cor_Diri}
\E\,\Phi(\zeta) \Phi(z) = 2G_{D}^{Diri}(\zeta,z).
\end{equation}
Although $\Phi$ cannot be realized as a random function, we write it as a (correlation functional-valued) function (or a Fock space field) since we are only interested in studying its correlation functions, which are well defined.
For example, in a cylinder $\mathcal C_r$, by \eqref{Diri Green cyl},
\begin{equation} \label{2pt cor_Diri cyl}
\E\,\Phi(\zeta) \Phi(z) = \log \left| \frac{\theta(r,\zeta-\bar{z})}{\theta(r,\zeta-z)} \right|^2 - 8 \pi^2 \frac{\im \zeta \cdot \im z}{r}.
\end{equation}	
In a general circular domain $\mathcal C$, it follows from \eqref{Diri Green v1} that 
\begin{equation}\label{2pt cor_MCD}
\E\,\Phi(\zeta)\Phi(z) = \log \left| \frac{z \, \omega(\zeta, \bar{z}^{-1})} {\omega(\zeta,z)} \right|^2 - 2 \Big \langle \bfs{h}(\zeta), \bfs{P}^{-1} \bfs{h}(z) \Big \rangle.
\end{equation} 

To work with the Gaussian free field on the Schottky double we first turn to the more general setting of the GFF on a compact Riemann surface. On a compact Riemann surface $M \equiv M_g$ of genus $g$, the Gaussian free field $\Psi$ is defined as a bi-variant Fock space field $(z,z_0) \mapsto \Psi(z,z_0)$, see \cite[Section 2]{KM17}. 
The correlation functions of $\Psi$ are given in terms of the bipolar Green's function as
\begin{equation*}
	\E\,\Psi(\zeta,z)\Psi(\ti{\zeta},\ti{z}) = 2\Big( G_{\zeta,z}(\ti{\zeta})-G_{\zeta,z}(\ti{z})\Big)=2\Big( G_{\ti{\zeta},\ti{z}}(\zeta)-G_{\ti{\zeta},\ti{z}}(z)\Big).
\end{equation*}
For example, when $g=1$, by \eqref{BH g=1}, 
\begin{equation}\label{2pt cor_g=1}
	\E\,\Psi(\zeta,z)\Psi(\ti{\zeta},\ti{z}) = \log \left| \frac{\theta(\ti{\zeta}-z)\theta(\ti{z}-\zeta)}{\theta(\ti{\zeta}-\zeta)\theta(\ti{z}-z)} \right|^2 - 4\pi \frac{\im(\zeta-z)\cdot\im (\ti{\zeta}-\ti{z})}{\im \tau}.
\end{equation}
In general, by \eqref{BG HG}, 
\begin{align}
\begin{split}
\label{2pt cor_HG}
\E\,\Psi(\zeta,z)\Psi(\ti{\zeta},\ti{z}) &= \log \left| \frac{\Theta(\mathcal A(\ti{\zeta})- \mathcal A(z)-e) \Theta(\mathcal A(\ti{z})- \mathcal A(\zeta)-e)}{\Theta(\mathcal A(\ti{\zeta})- \mathcal A(\zeta)-e) \Theta(\mathcal A(\ti{z})- \mathcal A(z)-e)} \right|^2
\\
&\quad -4\pi \Big \langle (\im \tau)^{-1} \im [\mathcal A(\zeta)- \mathcal A (z)], \im [\mathcal A(\ti{\zeta})- \mathcal A (\ti{z})] \Big \rangle. 
\end{split}
\end{align}

\subsection{Schottky double construction of GFFs} \label{Subsec_Schottky double}

In this subsection, we explain that GFF in a multiply connected domain can be constructed from GFF on a compact Riemann surface via the Schottky double. 
For the case $g=0$, such a construction was extensively studied in \cite{KM21}. As a pre-cursor to the general setting we give the standard example for the case $g=1$.

\begin{example}{GFF on a cylinder via the GFF on a torus}
Recall \eqref{eq: cylinder_representation} for the representation of a cylinder $\CC_r$. The Schottky double of the cylinder $\CC_r$ is a complex torus $\mathbb{T}_{r}:=\C/\Lambda$, where $\Lambda=\mathbb{Z}+\frac{ir}{\pi} \mathbb{Z}$.
Note that by \eqref{2pt cor_Diri cyl} and \eqref{2pt cor_g=1},
\begin{equation*}
\E\,\Phi(\zeta) \Phi(z) = \log \left| \frac{\theta(r,\zeta-\bar{z})}{\theta(r,\zeta-z)} \right|^2 -8 \pi^2 \frac{\im \zeta \cdot \im z}{r} = \frac{1}{2}\, \E\,\Psi(\zeta,\bar{\zeta})\Psi(z,\bar{z}). 
\end{equation*}
Therefore we have the Schottky double relation 
\begin{equation*}
\Phi(z)=\frac{1}{\sqrt{2}}\,\Psi(z,\bar{z})
\end{equation*}
within correlations. 
\end{example}

To describe the Schottky doubly relation on higher genus surfaces, it is convenient to take the reference domain as the circular domain $\CC=\CC(\bfs{r}, \bfs{\delta})$ in \eqref{circles C_j}. 
Recall that $C'_j$ is the reflection of $C_j$ with respect to $C_0$, and that the image of $C'_j$ under $\theta_j$ is $C_j$. Moreover for $\zeta \in C'_j$, we have $\theta_j(\zeta) = \zeta^* = 1 / \overline{\zeta}$.

We will denote by $F$ the region of the complex plane that is exterior to the $2g$ circles $\{ C_j, C'_j | j= 1, \cdots, g \}$. 
This is called the \emph{fundamental region} associated with the Schottky group generated by the M\"{o}bius maps $\{ \theta_j | j=1,\cdots, g \}$. 
Here, the two halves of $F$ (one is the reflection of the other one with respect to the unit circle) can be interpreted as two sides of a symmetric compact Riemann surface $\wh{\mathcal{C}}$ known as the \emph{Schottky double} of $\mathcal C$. The bi-variant field $\Psi(z, z_0)$ exists on the Schottky double (see the last section) and can be used to define the Dirichlet GFF in the circular domain $\mathcal{C}$ in the following way.

\begin{prop}\label{prop:Schottky_Dirichlet_GFF}
Within correlations, the Dirichlet GFF on a circular domain $\mathcal{C}$ satisfies the Schottky double relation
\begin{equation*}
	\Phi(z) =\frac{1}{\sqrt{2}}\,\Psi(z,z^*).
\end{equation*}
\end{prop}

Formally speaking, the idea behind Proposition \ref{prop:Schottky_Dirichlet_GFF} is just an instance of the method of images. This makes the statement intuitively clear, but we prove it by verifying that the correlation functions of the two sides match. To accomplish this we independently compute the Dirichlet and ER Green's functions in a circular domain.

\begin{lem} \label{Lem_Green Riemann Theta}
In a circular domain $D$, we have 
\begin{equation} \label{ER Green v2}
 G_D^{ER}(\zeta,z) = \frac12 \log \left| \frac{\Theta(\mathcal A(z)- \mathcal A(\zeta^*)-e) \Theta(\mathcal A(z^*)- \mathcal A(\zeta)-e)}{\Theta(\mathcal A(z)- \mathcal A(\zeta)-e) \Theta(\mathcal A(z^*)- \mathcal A(\zeta^*)-e)} \right|,
\end{equation}
where $z^*=1/\bar{z}$ and 
\begin{equation}
\label{Diri Green v2}
 G_D^{Diri}(\zeta,z)=G_D^{ER}(\zeta,z)-\pi \Big\langle (\im \tau)^{-1} \im[\mathcal A(\zeta^*)
 - \mathcal A (\zeta)], \im [\mathcal A(z^*)- \mathcal A (z)] \Big\rangle.
\end{equation}
\end{lem}

To prove Lemma \ref{Lem_Green Riemann Theta} first requires a few basic facts about the period matrices $\bfs P$ and $\tau$. We explain this relation next, but on first reading it is enough to simply accept \eqref{Imtau-P} as fact and return to this argument later.

On the Schottky double $\wh{\mathcal{C}}$ take the circle $C_j$ (which is identified with $C_j'$) as a $j$-th $a$-cycle of $\wh{\mathcal C}$. 
Also, for $\zeta \in C'_j$, take any line joining $\zeta$ and $\theta_j(\zeta)=\zeta^*$ as a $j$-th $b$-cycle. 
Let $\{\eta_j | j=1,\cdots, g \}$ be the basis of the space of all holomorphic 1-differentials on $\wh{\mathcal C}$ which is uniquely determined by \eqref{vj def}.
We now recall $\eta_j = dv_j$ and express the \emph{integral of the first kind} $\{v_j \, | \, j=1,\cdots, g \}$ in terms of harmonic functions $\{ h_j(\zeta) \, | \, j=0,1,\cdots, g \}$ given by \eqref{harmonic msrs}.
Recall that the period matrix $\bfs{P}=(P_{jk})_{j,k=1}^g$ of $h_j$ is given by \eqref{Pkj}. 
Since $h_j$ vanishes on $C_0$, one can extend $h_j$ to the fundamental region $F$ as 
\begin{equation*}
\wt{h}_j(\zeta)=
\begin{cases}
h_j(\zeta) &\text{on} \quad \overline{\mathcal{C}} 
\smallskip 
\\
-h_j(\zeta^*) &\text{on} \quad F \backslash \mathcal{C}.
\end{cases}
\end{equation*}
Let $\wh{v}_j$ be the holomorphic function on $F$ whose imaginary part is $\wt{h}_j$. Note that
\[
\wh{v}_j (\zeta^*)= \overline{\wh{v}_j (\zeta)}, \quad \zeta \in \mathcal C. 
\]
Then it follows that 
\begin{equation*}
(v_1,\cdots,v_g)^T= -\frac{1}{2\pi} \bfs{P}^{-1} (\wh{v}_1, \cdots, \wh{v}_g)^T,
\end{equation*}
where $\bfs P$ is given by \eqref{Pkj}, see e.g. \cite{CM07}. 
Moreover, the period matrix $\tau=(\tau_{jk})_{j,k=1}^g$ with $
\tau_{jk}=\oint_{b_k} dv_j$ is related to $\bfs{P}$ as 
\begin{equation} \label{Imtau-P}
\tau=\frac{i}{\pi}\bfs{P}^{-1}, \qquad (\im \tau)^{-1}=\pi \bfs{P}.
\end{equation}

\begin{proof}[Proof of Lemma \ref{Lem_Green Riemann Theta}]
By the characteristic property of the Green's function, all we need to show is that for $\zeta \in C_j$,
\begin{align*}
\frac{1}{2} \log \left| \frac{\Theta(\mathcal A(z)- \mathcal A(\zeta^*)-e) \Theta(\mathcal A(z^*)- \mathcal A(\zeta)-e)}{\Theta(\mathcal A(z)- \mathcal A(\zeta)-e) \Theta(\mathcal A(z^*)- \mathcal A(\zeta^*)-e)} \right|
=\pi \left \langle (\im \tau)^{-1} \im [\mathcal A(\zeta^*)- \mathcal A (\zeta)] , \im[\mathcal A(z^*)- \mathcal A (z)] \right \rangle. 
\end{align*}
For any $\zeta \in \mathcal C$, we have
\begin{align}
\begin{split}
\mathcal{A}(\zeta^*)-\mathcal{A}(\zeta)&= \Big( \int_{\zeta}^{\zeta^*} dv_1, \cdots, \int_{\zeta}^{\zeta^*} dv_g \Big)^T \label{diffA}
=\Big(v_1(\zeta^*)-v_1(\zeta),\cdots, v_g(\zeta^*)-v_g(\zeta) \Big)^T 
\\
&=\Big(\overline{v_1(\zeta)}-v_1(\zeta),\cdots, \overline{v_g(\zeta)}-v_g(\zeta) \Big)^T
= -2i \, \im \Big(v_1(\zeta),\cdots, v_g(\zeta) \Big)^T 
\\
&= \frac{i}{\pi} \bfs{P}^{-1} \bfs{h}(\zeta) = \tau \, \bfs{h}(\zeta). 
\end{split}
\end{align}
The last equality used \eqref{Imtau-P}. In particular, for $\zeta \in C_j$, the definition of the $h_k$ gives that $\bfs{h}(\zeta) = e_j$ for $e_j=(0,\cdots,1, \cdots,0)^T$, and so
\[
\mathcal{A}(\zeta^*)-\mathcal{A}(\zeta) = \tau e_j.
\]
For $\zeta \in C_j$ this gives rise to 
\[
\pi \Big \langle (\im \tau)^{-1} \im [\mathcal A(\zeta^*)- \mathcal A (\zeta) ] , \im[\mathcal A(z^*)- \mathcal A (z)] \Big\rangle = \pi \Big\langle e_j,\im [ \mathcal A(z)-\mathcal{A}(z^*)] \Big\rangle.
\]
On the other hand, it follows from \eqref{perio R-theta} that 
\begin{align*}
\frac{\Theta(\mathcal A(z)- \mathcal A(\zeta^*)-e) }{\Theta(\mathcal A(z)- \mathcal A(\zeta)-e) }&=\frac{\Theta(\mathcal A(z)- \mathcal A(\zeta)-e - \mathcal{A}(\zeta^*)+ \mathcal{A}(\zeta)) }{\Theta(\mathcal A(z)- \mathcal A(\zeta)-e) }
\\
&= \frac{\Theta(\mathcal A(z)- \mathcal A(\zeta)-e-\tau e_j) }{\Theta(\mathcal A(z)- \mathcal A(\zeta)-e) }= \exp \Big[ -2\pi i \Big( \langle \mathcal A(z)-\mathcal A(\zeta) , e_j \rangle -\frac{1}{2}\tau_{jj} \Big) \Big]
\end{align*}
and 
\[
\frac{ \Theta(\mathcal A(z^*)- \mathcal A(\zeta)-e)}{\Theta(\mathcal A(z^*)- \mathcal A(\zeta^*)-e)}=\exp \Big[ 2\pi i \Big( \langle \mathcal A(z^*)-\mathcal A(\zeta) , e_j \rangle -\frac{1}{2}\tau_{jj} \Big) \Big].
\]
Therefore we obtain
\[
\frac{1}{2} \log \left| \frac{\Theta(\mathcal A(z)- \mathcal A(\zeta^*)-e) \Theta(\mathcal A(z^*)- \mathcal A(\zeta)-e)}{\Theta(\mathcal A(z)- \mathcal A(\zeta)-e) \Theta(\mathcal A(z^*)- \mathcal A(\zeta^*)-e)} \right| 
= \pi \Big\langle e_j,\im[ \mathcal A(z)-\mathcal{A}(z^*)] \Big\rangle,
\]
which completes the proof.
\end{proof}

\begin{rem}
As a consequence of Lemma~\ref{Lem_Green Riemann Theta}, we have the functional relations: 
\begin{equation*}
\log \left| \frac{z \, \omega(\zeta, \bar{z}^{-1})} {\omega(\zeta,z)} \right|^2 = \log \left| \frac{\Theta(\mathcal A(z)- \mathcal A(\zeta^*)-e) \Theta(\mathcal A(z^*)- \mathcal A(\zeta)-e)}{\Theta(\mathcal A(z)- \mathcal A(\zeta)-e) \Theta(\mathcal A(z^*)- \mathcal A(\zeta^*)-e)} \right|
\end{equation*}
and
\begin{equation*} 
\Big\langle \bfs{h}(\zeta) , \bfs{P}^{-1} \bfs{h}(z) \Big\rangle =\pi \Big\langle (\im \tau)^{-1} \im[\mathcal A(\zeta^*)- \mathcal A (\zeta)], \im [\mathcal A(z^*)- \mathcal A (z)] \Big\rangle.
\end{equation*}
\end{rem}

\begin{proof}[Proof of Proposition \ref{prop:Schottky_Dirichlet_GFF}]
We need to check 
\begin{equation} \label{S-D of two pt cor}
	 \E\,\Phi (\zeta) \Phi (z) = \frac{1}{2} \,\E\,\Psi(\zeta,\zeta^*) \Psi(z,z^*).
	\end{equation}
By definition, the identity \eqref{S-D of two pt cor} is equivalent to 
\begin{equation} \label{S-D of Green}
G^{Diri}_D(\zeta,z)= \frac{1}{2}( G_{\zeta,\zeta^*}(z)-G_{\zeta,\zeta^*}(z^*) ).
\end{equation} 
Then the desired identity follows from Lemma~\ref{Lem_Green Riemann Theta} and the discussion in Section \ref{sec:GFF_correlation}.
\end{proof}

\subsection{Chiral fields and OPE exponentials}
\label{Subsection_chiral}
Chiral fields are defined by taking the holomorphic derivative of $\Phi$ (which complexifies it) and integrating it along paths. More precisely, for a path $\gamma$ in $D$ we define 
\[
\Phi^+(\gamma) = \int_\gamma J, \qquad J = \pa \Phi
\]
and write $\Phi^+(\zeta,z)$ for the collection of $\Phi^+(\gamma)$, where $\gamma$ is a path from $z$ to $\zeta.$
Then $(\zeta,z)\mapsto\Phi^+(\zeta,z)$ is a bi-variant multivalued field with correlation
\[
\E\,\Phi^+(\zeta,\zeta_0)\Phi(z) = G_D^+(\zeta,z)-G_D^+(\zeta_0,z),
\]
where $G_D^+$ is the complex Green's function defined in Subsection~\ref{Subsection_setup}. Furthermore
\[
\E\,\Phi^+(\zeta,\zeta_0)\Phi^+(z, z_0) = \log \lambda(\zeta,\zeta_0;z,z_0) + \pi \langle(\Im\,\tau)^{-1}(\mathcal{A}(\zeta)-\mathcal{A}(\zeta_0)),(\overline{\mathcal{A}(z)-\mathcal{A}(z_0)})\rangle,
\]
where $\lambda$ is the (generalized) cross-ratio,
\[
\lambda(p,q;\tilde p,\tilde q) = \frac{\Theta(\mathcal{A}(\tilde p) - \mathcal{A}(q)-e)\Theta(\mathcal{A}(\tilde q) - \mathcal{A}(p)-e)}{\Theta(\mathcal{A}(\tilde p) - \mathcal{A}(p)-e)\Theta(\mathcal{A}(\tilde q) - \mathcal{A}(q)-e)}.
\]
It is well known that the multivalued function $\lambda(p,q;\tilde p,\tilde q)$ is independent of $e,$ see \cite[Section VII.6]{FK92}.
It can be rewritten in terms of the prime form as 
\[
\lambda(p,q;\ti p,\ti q) = \frac{E(\ti p,q)E(\ti q,p)}{E(\ti p,p) E(\ti q,q)}
\]
by using the uniqueness of the cross-ratio (e.g. see \cite{Poor92}) and the periodicity properties of the prime form below.
The (Schottky-Klein) prime form $E$ (see \cite{Fay92,Mumford84} for its definition and basic properties) on a compact Riemann surface is the skew-symmetric bi-differential with periodicity properties
\begin{align*}
E(p+a_j,q) &= E(p,q),\\
E(p+b_j,q) &= \exp\Big(-2\pi i \big(\frac{\tau_{jj}}2 + \int_p^q \omega_j \big)\Big)E(p,q)
\end{align*}
and asymptotic behavior
\[
E(p,q) = \frac{z-z'}{\sqrt{dz}\sqrt{dz'}}\Big(1+O\big((z-z')^2\big)\Big)
\]
in a local chart that contains both $p$ and $q$, where $z$ is the image of $p$ for that chart and $z'$ the image of $q$. Its conformal dimension is $-1/2$ with respect to each variable.

For calculations it is sometimes convenient to introduce a formal 1-point field $z \mapsto \Phi^+(z)$ such that $\Phi^+(\zeta,z) = \Phi^+(\zeta)- \Phi^+(z).$ The 1-point field $\Phi^+$ would then have correlation
\[
\E\,\Phi^+(\zeta)\Phi^+(z) = -\log E(\zeta,z) + \pi \langle(\Im\,\tau)^{-1} \mathcal{A}(\zeta),\overline{\mathcal{A}(z)})\rangle.
\]
By defining $\Phi^-= \overline{\Phi^+}$ we then obtain the decomposition
\[
\Phi=\Phiplus+\Phiminus,
\]
in the sense that all correlations against $\Phi$ can be replaced by correlations against $\Phiplus + \Phiminus$.

We now define the OPE exponential of the field $\Phi^+(\zeta, z)$. As explained in \cite{KM17, KM21}, the OPE exponential is a conformal field defined via a power series expansion of the exponential function, but for which multiplication is replaced by the operator product expansion. For $\sigma \in \R$ we denote the OPE exponential by 
\[
\OO\{\sigma\cdot z - \sigma\cdot z_0\} = e^{*i\sigma\Phi^+(z,z_0)}.
\]
See \cite{KM17, KM21} for the precise power series expansion. Each term in the expansion of the OPE exponential is an element of the Fock space family generated by $\Phi$, therefore so too is the sum. Although the OPE exponential appears to be more complicated than the Wick exponential of $\Phi^+(z, z_0)$, they turn out to be the same up to a factor of a deterministic conformal field. This is shown in \cite[Proposition 5.3]{KM17}, which we now recall.

\begin{prop}[Proposition 5.3 of \cite{KM17}]
\label{Prop 5.3 of KM17}
Let $X = \sum_j \tau_j \Phi(z_j)$ for a divisor $\bfs \tau$ satisfying $\int \bfs \tau = 0$. Then for $\alpha \in \C$
\[
e^{* \alpha X } = \exp \Big( \frac{\alpha^2}{2} \E[X * X] \Big) e^{\odot \alpha X}.
\]
\end{prop}

Thus the OPE exponential can be expressed entirely in terms of the Wick exponential and the zero order coefficient in the OPE expansion of the field $X$ with itself. The operator product expansion is computed by using Wick's formula to compute the expansion of 
\[
\Big( \sum_j \tau_j \Phi(z_j) \Big) \Big( \sum_j \tau_j \Phi(z_j + \epsilon_j) \Big),
\]
as the support points $\epsilon_j$ all converge to the origin. The same methods as in Proposition~\ref{Prop 5.3 of KM17} can be straightforwardly extended to apply this calculation to the field $\Phi^+(z, z_0)$, the only difference being that $\Phi^+(z, z_0)$ is specified if an integral of the current field is taken along a given path and so we must specify how the two paths $\gamma,\ti\gamma$ are related in the product
\[
\Phi^+(z, z_0)\Phi^+(\ti z, \ti z_0).
\]
We work with the natural choice that the endpoints $(\ti z, \ti z_0)$ approach $(z, z_0)$ while the paths $\gamma,\ti\gamma$ stay disjoint but close to each other. For this choice it can be shown by straightforward calculation that $c^+(z,z_0)$ depends only on the choice of $\gamma$ (as the prime form does, see \cite[Chapter 1]{Fay92}) and
\[
c^+(z,z_0) := \frac12 \, \E\, \Phi^+*\Phi^+(z, z_0) = \log E(z,z_0) +  \frac\pi2 \langle(\Im\,\tau)^{-1}(\mathcal{A}(z)-\mathcal{A}(z_0)),(\overline{\mathcal{A}(z)-\mathcal{A}(z_0)}) \rangle.
\]
Furthermore $c^+(z,z_0)$ is a PPS form of order $(-\frac12,0)$ with respect to $z$ and $z_0$. Therefore, by the same reasoning as in Proposition~\ref{Prop 5.3 of KM17} we have
\begin{align}\label{eq:two_point_OPE_exponential}
\OO\{\sigma\cdot z - \sigma \cdot z_0 \} = e^{-\sigma^2 c^+(z,z_0)} \, e^{\odot i \sigma \Phi^+(z,z_0) }.
\end{align}
As a conformal field the OPE exponential is a holomorphic differential with conformal dimension $\sigma^2/2$ at $z,z_0$, and satisfies
\[
\E\,e^{*i\sigma\Phi^+(z,z_0)} = e^{-\sigma^2 c^+(z, z_0)} = \frac1{E(z,z_0)^{\sigma^2}} \exp \left \{ -\frac{\pi}{2} \sigma^2 \langle (\Im\,\tau)^{-1}(\mathcal{A}(z)-\mathcal{A}(z_0)), (\overline{\mathcal{A}(z)-\mathcal{A}(z_0))} \right \}.
\]

Finally we introduce the OPE exponential of more complicated linear combinations of the fields $\Phi^{\pm}$. Recall that in the introduction we defined $\Phi \{ \bfs \tau \}$ as the linear combination
\[
\Phi \{ \bfs \tau \} = \Phi \{ \bfs \tau^+, \bfs \tau^- \} = \sum \tau_j^+ \Phi^+(z_j) - \tau_j^- \Phi^-(z_j),
\]
for any double divisor $\bfs \tau= ( \bfs \tau^+, \bfs \tau^- )$ ($\bfs\tau^{\pm} = \sum \tau_j^{\pm} \cdot z_j$) satisfying the neutrality condition 
\[
\int \bfs \tau^+ + \bfs \tau^- = 0.
\]
The OPE exponential $\OO\{\bfs \tau\}$ is then defined by 
\[
\OO\{ \bfs \tau \} := e^{ \ast i \Phi\{ \bfs \tau \} }.
\]
Although these OPE exponentials depend on many more points than $\OO \{ \sigma \cdot z - \sigma \cdot z_0 \}$, it can be shown that
\begin{align}\label{eq:OPE_product_formula}
\OO\{\bfs\tau\} = \OO^{(\tau_1^+)}(z_1)\overline{\OO^{(\overline{\tau_1^-})}(z_1)}\cdots\OO^{(\tau_n^+)}(z_n)\overline{\OO^{(\overline{\tau_n^-})}(z_n)},
\end{align}
where $\OO^{(\tau)}(z) = \OO\{\tau\cdot z - \tau\cdot z_0\}$. The result does not depend on the choice of a base point $z_0$ due to the neutrality condition. This representation shows that the multiplication rule
\[
\OO\{ \bfs \tau_1 \}\OO\{ \bfs \tau_2 \} = \OO\{ \bfs \tau_1 + \bfs \tau_2 \}.
\]
holds. This shows that while the fields $\OO \{ \bfs \tau \}$ may appear to be more complicated, they can be entirely expressed in terms of products of the simpler field $\OO \{ \sigma \cdot z - \sigma \cdot z_0 \}$. Products of the latter fields can be computed using equation \eqref{eq:two_point_OPE_exponential} and the Wick rule
\[
e^{\odot \alpha X} e^{\odot \beta Y} = e^{\alpha \beta \operatorname{Cov}(X,Y)} e^{\odot (\alpha X + \beta Y)}.
\] 
For a more detailed discussion of OPE exponentials in a simply connected domain see \cite{KM21}.

\subsection{Background charge modification} \label{Subsection_background}

To implement the $c < 1$ version of our conformal field theory we introduce the background charge modification. It is simplest to introduce on the Schottky double $\wh{\CC}$ and then transfer the results back to the circular domain $\CC$. The presentation here is a specialization of the results found in \cite[Section 6]{KM17}. We remind the reader of the following basic facts:

\begin{itemize}
\item The surface $\wh{\CC}$ is covered by a collection of charts $(\phi, U)$, $U \subset \wh{\CC}$ open, where $\phi : U \to \C$. The transition map between overlapping charts is always assumed to be holomorphic. A non-random field is an assignment of a smooth function to each chart. For a non-random field $\psi$ and charts $\phi, \wt{\phi}$, we write $\psi=(\psi \, \| \, \phi)$ and $\wt{ \psi }= ( \psi \, \| \, \wt{\phi} ) $ for the assignment of $\psi$ to each of the charts. 
\smallskip 
 \item A non-random field $\psi$ is called a PPS$(\mu,\nu)$ form if $\psi$ satisfies the transformation law
\begin{equation*}
\psi=\tilde{\psi} \circ h+ \mu \log h'+\nu \log \overline{ h' },
\end{equation*}
where $h$ is the transition map between two overlapping charts $\phi, \wt{\phi}$. 

\smallskip 
\item A non-random field $\psi$ is called a harmonic form if $\psi = \psi^+ + \psi^-$ for some (multivalued) holomorphic form $\psi^+$ and (multivalued) anti-holomorphic form $\psi^-$. 
\smallskip 
\item A (harmonic) PPS form $\psi$ is simple if $\pa \bp \psi$ is a finite linear combination of Dirac delta measures. 
\smallskip 
\item For a given point $q \in \wh{\CC}$, we call $\psi_q$ a basic form if it is a simple form with a sole singularity at $q$, 
\begin{equation*}
\psi_q(z) \sim \log \frac{1}{|z-q|^2}, \qquad \textup{equivalently} \qquad \pa \bp \psi_q = -\pi \delta_q, 
\end{equation*}
in any chart. 
\smallskip
\end{itemize}

Simple PPS forms have the following properties. These results are proved in Theorem 6.4, Proposition 6.6, and Corollaries 6.3 and 6.5 of \cite{KM17}.

\begin{prop}
Throughout let $M$ be a compact Riemann surface with genus $g \neq 1$ and $\chi = 2 - 2g$. Then
\begin{enumerate}
 \item For any $q \in M$ there exists a unique basic form $\psi_q$ (unique up to an additive constant), and it is a PPS$(1/\chi, 1/\chi)$ form.
 \smallskip 
 \item A finite linear combination of basic forms is a simple form.
 \smallskip 
 \item Let $\psi$ be a simple PPS$(ib, ib)$ form on a compact Riemann surface of genus $g$. Then 
 \begin{equation}\label{NC b}
 \frac{i}{\pi} \int \pa \bp \psi = b(2-2g) = b \chi.
 \end{equation}
 \item Suppose that $g \not= 1$ and a finite atomic measure $\bfs{\beta}=\sum \beta_k \delta_{q_k}$ on the surface satisfies $\sum \beta_k = b(2-2g)$. Then there is a unique (up to an additive constant) simple PPS($ib,ib$) form $\psi_{\bfs{\beta}}$ satisfying 
 \begin{equation} \label{psi beta from beta}
 \bfs{\beta}=\frac{i}{\pi} \pa \bp \psi_{ \bfs{\beta} } .
 \end{equation}
 \item Given $q \in M$ and $\bfs \beta$ of part (iv), any $\psi_{\bfs \beta}$ in part (iv) satisfies
 \[
 \psi_{\bfs \beta}(z) - \psi_{\bfs \beta}(z_0) = ib \chi (\psi_q(z) - \psi_q(z_0)) + 2 i \sum_k \beta_k (G_{q_k, q}(z) - G_{q_k, q}(z_0)),
 \]
 and the right hand side does not depend on the choice of $q$.
\end{enumerate}
\end{prop}

In the case $g=1$ the last formula specializes to
\[
\psi_{\bfs \beta}(z) - \psi_{\bfs \beta}(z_0) = 2ib \log |\omega(z)| -2ib \log |\omega(z_0)|   + 2 i \sum_k \beta_k (G_{q_k, q}(z) - G_{q_k, q}(z_0)),
\]
where $\omega$ is a holomorphic form (equivalently, a $(1,0)$ differential). 

Given a simple PPS($ib,ib$) form $\psi$, the corresponding background charge $\bfs \beta $ is defined by
\begin{equation*} 
\bfs{\beta}=\frac{i}{\pi} \pa \bp \psi= \sum \beta_k \delta_{q_k}= \sum \beta_k \cdot q_k.
\end{equation*}
Here, we think of $\bfs \beta$ as a measure. In the algebraic geometry literature it would often be called a divisor. Equation \eqref{NC b} is referred to as the \textit{neutrality condition} of the simple form $\psi$, it is also called the neutrality condition of the corresponding background charge $\bfs \beta$. Given the background charge $\bfs \beta$, we will typically define $b$ as the unique constant such that
\begin{align}\label{NC b2}
\int \bfs \beta = \sum_k \beta_k = b (2 - 2 g) = b \chi.
\end{align}

We write $\psi^+_q, \psi^-_q$ for holomorphic part and anti-holomorphic part of a basic simple form $\psi_q$. This gives the decomposition 
\[
\psi_q = \psi^+_q + \psi^-_q.
\]
Let $\bfs \beta^\pm$ be background charges (atomic divisors), which are allowed to satisfy different neutrality conditions 
\[
\sum_k \beta_k^\pm = b^\pm \chi,
\]
although in the rest of this paper they will always be the same. The divisors $\bfs \beta^{\pm}$ keep track of modifications of the field $\Psi$ via linear combinations of the basic forms $\psi_q^{\pm}$. More precisely, we define the field $\Psi_{\bfs \beta^+, \bfs \beta^-}(z, z_0)$ via
\begin{equation*}
\Psi_{\bfs \beta^+, \bfs \beta^-}(z, z_0) = \Psi(z, z_0) + i \sum_k \beta_k^+ \left( \psi_{q_k}^+(z) - \psi_{q_k}^+(z_0) \right) - i \sum_k \beta_k^- \left( \psi_{q_k}^-(z) - \psi_{q_k}^-(z_0) \right).
\end{equation*}

Using these definitions for the fields on a compact Riemann surface, we now define the field $\Phi_{\bfs \beta}$. The idea is to define $\Phi_{\bfs \beta}$ and its mean $\varphi_{\bfs \beta}$ as a ``symmetric'' case of $\Psi_{\bfs \beta^+, \bfs \beta^-}$ on the Schottky double. The special case that we use is simply $\bfs \beta^+ = \bfs \beta^- = \bfs \beta$, which is an example of the method of images. We assume that the background charge $\bfs\beta$ is supported on $C_0$, and all charges in $\bfs \beta$ are real (i.e. $\bfs\beta(q_k) \in \R$ for all $k$). We now define 
\begin{equation}
\Phi_{\bfs{\beta}}= \Phi + \varphi_{\bfs{\beta}}, \qquad \varphi_{\bfs{\beta}}(z) = \E\, \Psi_{\bfs\beta,\bfs\beta}(z,z^*) = -2 \sum_k \beta_k\, \Im (\psi_{q_k}^+(z) - \Psi_{q_k}^+(z^*)). 
\end{equation}
This definition makes $\Phi_{\bfs{\beta}}$ a PPS($ib,-ib$) form with $ \varphi_{ \bfs{\beta} }=\E\,\Phi_{\bfs{\beta}}$. An equivalent way of defining it is
\[
\Phi_{\bfs \beta}(z) = \frac{1}{\sqrt{2}}\,\Psi_{\sqrt{2} \bfs \beta, \sqrt{2} \bfs \beta}(z, z^*),
\]
which makes the method of images representation more transparent, and matches the representation in Proposition \ref{prop:Schottky_Dirichlet_GFF}. In a multiply connected domain $D$, this leads to the formula presented in equation \eqref{Phi beta}, i.e.
\[
\varphi_{\bfs \beta}(z) = - 2b \arg w'(z) - 2 \sum \beta_k\, \Im\, G^+_D(w(q_k),w(z)), 
\]
where $w : (D, C_0, q) \to (\H_g, \R, \infty)$ is a conformal map satisfying a hydrodynamic normalization near $w(q) = \infty$. 
Here, we recall that $G_D^+$ is the complex Green's function. 
In particular, in a circular domain, it follows from Lemma~\ref{Lem_Green Riemann Theta} that 
\begin{align*}
G_D^+(z,z_0) & = \frac14 \log \left(   \frac{\Theta(\mathcal A(z)- \mathcal A(1/\bar{z}_0)-e) \Theta(\mathcal A(1/z)- \mathcal A(\bar{z}_0)-e)}{\Theta(\mathcal A(z)- \mathcal A(z_0)-e) \Theta(\mathcal A(1/z)- \mathcal A(1/z_0)-e)}    \right)  
\\
&\quad + \frac{ \pi i}{2} \Big\langle (\im \tau)^{-1} ( \mathcal A(z)+\mathcal A(1/z) ) ,  \im [\mathcal A(1/\bar{z}_0)- \mathcal A (z_0)] \Big\rangle
\end{align*} 
up to periods. 
In addition, we define formal 1-point fields $\Phi_{\bfs\beta}^\pm$ as 
\[
\Phi_{\bfs\beta}^\pm= \Phi^\pm + \varphi_{\bfs\beta}^\pm,
\]
where 
$\varphi_{\bfs\beta}^+(z) = \psi^+_{\bfs{\beta}}(z) - \psi^-_{\bfs \beta}(z^*)$,
and $\varphi_{\bfs\beta}^-(z) = \psi^-_{\bfs \beta}(z) - \psi^+_{\bfs \beta}(z^*)$. The current field $J_{ \bfs{\beta} }$ is defined as 
\begin{equation*} 
J_{ \bfs{\beta} } := \pa \Phi_{ \bfs{\beta} }=J+\jmath_{ \bfs{\beta} }, \qquad
\jmath_{ \bfs{\beta} }(z)= i \sum \beta_k \pa_z (G_D^+(q_k,z)-G_D^-(q_k,z)). 
\end{equation*} 
Based on the definition of the current field we have the following.

\begin{prop}
The field $\Phi_{ \bfs{\beta} }$ has a stress tensor 
\begin{equation*}
A_{ \bfs{\beta} }:= A_{ (0) }+\Big(ib\pa-\jmath_{ \bfs{\beta} } \Big)J, 
\end{equation*}
and its Virasoro field is given as
\begin{equation*}
T_{ \bfs{\beta} } :=-\dfrac{1}{2}J_{ \bfs{\beta} } \ast J_{ \bfs{\beta} }+ib\,\pa J_{ \bfs{\beta} }. 
\end{equation*}
\end{prop}
\begin{proof}
The proof is the same as the one in \cite[Proposition 3.5]{BKT23} and \cite[Theorem 5.2]{KM21}.
\end{proof}

Given the field $\Phi_{\bfs \beta}$ we can also consider its OPE exponential. In \cite[Lemma 5.2]{KM17} it is proved that for a smooth function $f$ and a mean zero Gaussian conformal field $X$ (say $X = \Phi(z) - \Phi(z_0)$) that
\[
e^{*(X + f)} = e^f e^*X.
\]
Consequently, if we denote the OPE exponential of $\Phi_{\bfs \beta}^+(z, z_0)$ by
\[
\OO_{\bfs\beta}\{\tau\cdot z - \tau\cdot z_0\} = e^{*i\tau\Phi_{\bfs\beta}^+(z,z_0)}\]
then we have the formula
\[
\OO_{\bfs\beta}\{\tau\cdot z - \tau\cdot z_0\} = e^{i\tau(\varphi_{\bfs\beta}^+(z)-\varphi_{\bfs\beta}^+(z_0))} \OO\{\tau\cdot z - \tau\cdot z_0\}.
\]
The last formula combined with the OPE product formula \eqref{eq:OPE_product_formula} us to compute the OPE exponential of $\Phi_{\bfs \beta} \{ \bfs \tau \}$ for a double divisor $\bfs\tau = (\bfs\tau^+,\bfs\tau^-), \bfs \tau^\pm = \sum_{j=1}^n \tau^\pm_j\cdot z_j$ with neutrality condition $\int \bfs \tau^+ + \bfs \tau^- = 0$. Recalling that $\OO_{\bfs \beta} \{ \bfs \tau \} = e^{*i \Phi_{\bfs \beta} \{ \bfs \tau \}}$, the precise formula is that    
\[
\OO_{\bfs\beta}\{\bfs\tau\} = \OO_{\bfs\beta}^{(\tau_1^+)}(z_1)\overline{\OO_{\bfs\beta}^{(\overline{\tau_1^-})}(z_1)}\cdots\OO_{\bfs\beta}^{(\tau_n^+)}(z_n)\overline{\OO_{\bfs\beta}^{(\overline{\tau_n^-})}(z_n)},
\]
where we set $\OO_{\bfs\beta}^{(\tau)}(z) = \OO_{\bfs\beta}\{\tau\cdot z - \tau\cdot z_0\}$. As before this representation does not depend on the choice of the base point $z_0$.
If the support of $\bfs\tau$ contains the point at infinity, or if the support of $\bfs \tau$ intersects the support of $\bfs\beta$, then a rooting procedure or a rescaling procedure is needed. See \cite[Section~12.3]{KM13} or \cite{KM21} for more details.

\section{Ward's and BPZ equations} \label{Section_Ward and BPZ}

In this section we prove Theorem \ref{Thm_Ward}, which establishes the Eguchi-Ooguri version of Ward's equation in canonical multiply connected domains, and Theorem \ref{Thm_BPZ}, which takes an OPE limit of the Ward equation to obtain the BPZ and BPZ-Cardy equation. The latter makes use of the level two degeneracy property of vertex exponentials, and is the key step in establishing the connection between conformal field theory and SLE.

\subsection{Ward's identity}

We begin by introducing the global Ward's functional, which plays a key role in understanding the relationship between stress tensors and vector fields. 
To define Ward's functional, we consider a meromorphic vector field $v$ defined on $\overline{D}$, which is continuous up to its boundary. 
Ward's functionals $W^{\pm}_{  \epsilon  }(v)$ are then given by
\begin{equation} \label{W A relation}
W^{+}_{ \epsilon }(v)= W^+(v;  D_\epsilon  )	=	\frac{1}{2\pi i}\int_{\pa  D_\epsilon  } vA_{ \bfs \beta } -\frac{1}{\pi} \int \!\! \int_{  D_\epsilon} (\bp v)A_{ \bfs \beta }, \qquad	W^-_{   \epsilon}  (v)=\overline{W^+_{ \epsilon } (v)},
\end{equation}
where $D_\epsilon= D \setminus \bigcup B(q_j, \epsilon)$. Note here that we have excluded the nodes of $\bfs \beta$, where the integrand $v A_{ \bfs \beta }$ has singularities.

Since $v$ is meromorphic the term $\bp v$ is zero except at the poles of $v$, and so $\bp v$ is interpreted in the sense of distributions. For a chordal standard domain $D$, the boundary $\pa D$ is the real axis together with a collection of contours encircling each horizontal slit . 
By \cite[Proposition 5.9]{KM13}, it follows that for $X$ in the OPE family $\mathcal{F}_{ \bfs \beta }$,
\begin{equation} \label{L W relation}
\LL_{v}^+ X = W^+_{  \epsilon  } X, \qquad \LL_{v}^- X = W^-_{  \epsilon   } X. 
\end{equation}
Here, the Lie derivatives do not act on the $q_k$ variables.
The next lemma gives a representation of the stress tensor in terms of the Ward's functional. In the next lemma and the rest of the paper we will frequently use that for a chordal standard domain $D$
\[
\overline{\int_{\pa D} f(\xi) \, d \xi} = \int_{\pa D} \overline{f(\xi)} \, d \xi,
\]
which follows because $\pa D$ consists of $\R$ and a collection of horizontal slits.

\begin{lem}\label{lem: stress_tensor_representation}
In the identity chart of a chordal standard domain $D$, we have 
\begin{equation*}
2	A_{ \bfs \beta } (\zeta) = W^+(v_\zeta)+W^- (v_{\bar{\zeta}})- \frac{1}{2\pi i} \int_{\pa D  \setminus \R }   \Big( v_\zeta(\xi)-v_\zeta( \bar{\xi}) \Big) A_{ \bfs \beta }(\xi)\,d\xi +   \frac{1}{2\pi i} \sum_k \oint_{ (q_k) } v_\zeta(\xi) A_{ \bfs \beta }(\xi)     
\end{equation*}
within correlations. Here, $\pa D  \setminus \R $ is a combination of small contours encircling each slit of $D$. 
\end{lem}

\begin{proof} 
For $\zeta \in D,$ we define the reflected vector field
\[
v_\zeta^\sharp(z):=\overline{v_\zeta(\bar{z}) }=v_{\bar{\zeta}}(z).
\]
By using the definition \eqref{W A relation}, we have
\[
W^+_{   \epsilon }(v_\zeta)=- \frac{1}{\pi}\int_{D_{   \epsilon } }(\bp v_\zeta)	A_{ \bfs \beta }+ \frac{1}{2\pi i}\int_{\pa D_{   \epsilon } } v_{\zeta}	A_{ \bfs \beta },
	\qquad W^+_{  \epsilon }(v_{\zeta}^{\sharp})=\frac{1}{2\pi i}\int_{\pa D_{  \epsilon }}v_\zeta^{\sharp}	A_{ \bfs \beta }.
\]
The second equation uses that $v_\zeta^\sharp$ is holomorphic on $D$, which follows from $v_\zeta$ being meromorphic on $D$ with sole pole at $\zeta$. 
Since $	A_{ \bfs \beta }(\bar{\xi}) =\bar{	A }_{ \bfs \beta }(\xi)$ on $D,$ we have 
\begin{align*}
W^-_{\epsilon}(v_{\bar{\zeta}}) &  = \overline{W^{+}_{ \epsilon }(v_{\zeta}^{\sharp})} = \overline{\frac{1}{2\pi i}\int_{\pa D   \setminus \R   } v_{\bar{\zeta}} (\xi )	A_{ \bfs \beta }(\xi)\, d\xi}  + \frac{1}{2\pi i} \sum_k \int_{ \partial B(q_k,\epsilon) \cap \mathbb{H} } v_\zeta( \bar{\xi} ) A_{ \bfs \beta }(\bar{\xi}) \,d\xi 
 \\
&=	-\frac{1}{2\pi i}\int_{\pa D  \setminus \R  } v_\zeta (\bar{\xi} )	A_{ \bfs \beta }(\xi)\,d\xi  - \frac{1}{2\pi i} \sum_k \int_{ \partial B(q_k,\epsilon) \cap \mathbb{H}_- } v_\zeta( \xi ) A_{ \bfs \beta }( \xi) \,d\xi,  
\end{align*}
where $\mathbb{H}_-:=\{ \xi :\im \xi <0 \}$ and, by convention, the curve $\partial B(q_k, \epsilon) \cap \mathbb{H}_{-}$ is oriented in the counterclockwise direction.
Then it follows from $\bp v_\zeta=-2\pi \delta_\zeta$ that 
	\begin{align*}
	W^+(v_\zeta)+W^- (v_{\bar{\zeta}})
	&=- \frac{1}{\pi}\int_{D}(\bp v_\zeta)	A_{ \bfs \beta }+\frac{1}{2\pi i} \int_{\pa D  \setminus \R   } \Big( v_\zeta(\xi)-v_\zeta( \bar{\xi}) \Big) A_{ \bfs \beta }(\xi)\,d\xi   - \frac{1}{2\pi i} \sum_k \oint_{ (q_k) } v_\zeta(\xi) A_{ \bfs \beta }(\xi) \, d \xi   
	 \\
	&= 2	A_{ \bfs \beta }(\zeta) +\frac{1}{2\pi i} \int_{\pa D  \setminus \R   } \Big( v_\zeta(\xi)-v_\zeta( \bar{\xi}) \Big) A_{ \bfs \beta }(\xi)\,d\xi   - \frac{1}{2\pi i} \sum_k \oint_{ (q_k) } v_\zeta(\xi) A_{ \bfs \beta }(\xi) \, d \xi   .
	\end{align*}
\end{proof}

Together, Lemma \ref{lem: stress_tensor_representation} and equation \eqref{L W relation} complete the proof of equation \eqref{eq: Ward_w_integral_term} in Theorem \ref{Thm_Ward}. Equation \eqref{Ward eq:residue form} in Theorem \ref{Thm_Ward} is proved in \cite[Lemma 7.2]{KM21}. Although the statement in \cite{KM21} is originally formulated for the $g=0$ case, the proof applies to our current setup as well since it pertains to a local property independent of the global underlying geometry.  We restate the lemma here for the reader's convenience.

\begin{lem}[Lemma 7.2 of \cite{KM21}]  \label{Lem_Ward residue q_k}
We have 
\begin{equation}
\E\oint_{ (q_k) } v_\zeta A_{ \bfs \beta }X = v_\zeta (q_k) \partial_{q_k} \E\,X,
\end{equation}
where $X$ is a tensor product of fields in $\FF_{\bfs \beta}$ whose nodes do not contain any of the $q_k$ points in $\bfs \beta$. 
\end{lem}

To complete the proof of Theorem 1.1 only requires showing the generalized Eguchi-Ooguri equation \eqref{EO main statement}. This is done in the next section. Once the generalized Eguchi-Ooguri equations are proved, the final equation \eqref{eq: Ward_no_integral} follows directly from \eqref{eq: Ward_w_integral_term}~--~\eqref{Ward eq:residue form}.

\subsection{Generalized Eguchi-Ooguri equation}

In this subsection we give the proof of the generalized Eguchi-Ooguri equations \eqref{EO main statement}, thereby completing the proof of Theorem \ref{Thm_Ward}. We start by showing that \eqref{EO main statement} holds for the OPE family $\FF_{(b)}$, i.e. when the background charge modification only happens at one point, and then extend it to the case of general background charge for a divisor $\bfs \beta$ that satisfies \eqref{NC b2}.

\begin{lem} \label{Lem_EO beta0}
For any $\XX \in \FF_{(b)}$, the generalized Eguchi-Ooguri equation 
\begin{equation} \label{EO eq v1}
\frac{1}{2\pi i} \int_{\pa D  \setminus \R}   \Big( v_\zeta(\xi)-v_\zeta( \bar{\xi}) \Big) \E A(\xi) \XX \,d\xi +\nabla_{ \mathbb{H}_g } \E \XX=0
\end{equation}
holds in the identity chart of a chordal standard domain $D$.  
\end{lem} 

\begin{proof}
First recall that $\FF_{(b)}$ is the OPE family by generated by $\Phi_{(b)}$, which is simply $\Phi_{\bfs \beta}$ in the special case $\bfs \beta = b \chi \cdot \infty$. Placing the background charge at infinity means that $\Phi_{(b)} = \Phi$ in the standard identity chart.

Let $\mathcal{W}$ be the family of conformal fields that satisfy the equation \eqref{EO eq v1}. 
In order to prove that $\FF_{(b)} \subset \mathcal{W}$, it suffices to show that any tensor product of GFFs and their derivatives is an element of $\mathcal{W}$, and that the family $\mathcal{W}$ is closed under OPE products. We break the proof into several parts.

\subsubsec{Proof that $\Phi_{(b)}(z) \in \mathcal{W}$}
We first show that in the identity chart the identity
\begin{equation} \label{EO 2pt}
\frac{1}{2\pi i}\int_{\partial D \setminus \R} v_{\zeta} (\xi) \E A(\xi)\XX\,d\xi+\overline{\frac{1}{2\pi i}\int_{\partial D \setminus \R} v_{\bar{\zeta}}(\xi) \E A(\xi) \XX \,d \xi }+\nabla_{ \mathbb{H}_g } \E \XX=0
\end{equation}
holds for $\XX=\Phi_{(b)}(z)$ and $A = A_{(b)}$. Note that we wrote the term $\nabla_{\mathbb{H}_g} \E \XX$ to show the similarity with \eqref{EO eq v1}, but that in this case $\E \XX = \E\,\Phi_{(b)} = 0$ so that term can be ignored.
For the remainder, recall that in the identity chart we have $j_{(b)} \equiv 0$ and so
\begin{align*}
A \equiv A_{(b)}= -\frac{1}{2}J \odot J + ib \pa J.
\end{align*}
Then by Wick's formula, we have
\begin{align} \label{eq: AbPhib_Wick}
\E A_{(b)}(\xi) \Phi_{(b)}(z) = ib \pa_\xi  \E J(\xi) \Phi(z) = 2ib \pa_\xi^2 G_D(\xi,z). 
\end{align}
By residue calculus, and using the asymptotic $\zeta \to z$ expansions
\begin{equation*}
v_\zeta(z)=\frac{2}{\zeta-z}+O(1), \qquad \pa_{ \zeta } G_D(\zeta,z)=-\frac12 \frac{1}{\zeta-z}+O(1),
\end{equation*}
we can evaluate the contour integrals as
\begin{align*}
\frac{1}{2\pi i } \int_{\partial D \setminus \R} v_\zeta(\xi) \pa_\xi^2 G_D(\xi,z) \,d\xi = -2 \pa_\zeta^2 G_D(\zeta,z)+ \frac12 v_\zeta'(z), \qquad
\frac{1}{2\pi i } \int_{\partial D \setminus \R} v_{\bar{\zeta}}(\xi) \pa_\xi^2 G_D(\xi,z) \,d\xi = \frac12 v_{ \bar{\zeta} }'(z). 
\end{align*}
By combining the last two formulas with the Wick expansion \eqref{eq: AbPhib_Wick} we obtain
\begin{align*}
\frac{1}{2\pi i}\int_{\partial D \setminus \R} v_{\zeta} (\xi) \E A_{(b)}(\xi) \Phi_{(b)}(z) \,d\xi &= \frac{b}{\pi} \int_{\partial D \setminus \R} v_{\zeta}(\xi) \pa_\xi^2 G_D(\xi,z) \,d \xi = bi \Big( -4 \pa_\zeta^2 G_D(\zeta,z)+ v_\zeta'(z) \Big),
\\
\frac{1}{2\pi i}\int_{\partial D \setminus \R} v_{\bar{\zeta}} (\xi) \E A_{(b)}(\xi) \Phi_{(b)}(z) \,d\xi &= \frac{b}{\pi} \int_{\partial D \setminus \R} v_{\bar{\zeta}}(\xi) \pa_\xi^2 G_D(\xi,z) \,d \xi = bi \, v_{\bar{\zeta}}'(z) . 
\end{align*}
Therefore the identity \eqref{EO 2pt} for $\XX = \Phi_{(b)}(z)$ is equivalent to the identity
\begin{equation}\label{eq: Greens_function_crucial_identity}
4 \pa_\zeta^2 G_D(\zeta,z) = v_\zeta'(z) - \overline{ v_{\bar{\zeta}}'(z) } . 
\end{equation}
This identity is proved separately in Lemma~\ref{lem: Green_crucial_identity} below. 

\subsubsec{Proof that $\Phi(z_1)\Phi(z_2) \in \mathcal{W}$}
We first compute the correlation $\E A_{(b)}(\xi) \Phi_{(b)}(z_1) \Phi_{(b)}(z_2)$ as
\[
\E A_{(b)}(\xi) \Phi_{(b)}(z_1) \Phi_{(b)}(z_2) = -\tfrac{1}{2} \E (J \odot J)(\xi) \Phi(z_1) \Phi(z_2) + ib \E \partial J(\xi) \Phi(z_1) \Phi(z_2),
\]
which uses that $\Phi_{(b)} \equiv \Phi$ in the identity chart. The last term is zero since the fields $\partial J$ and $\Phi$ are mean zero. For the first term we have by Wick's formula that
\begin{align*}
\E A(\xi) \Phi(z_1)\Phi(z_2) = - \E J(\xi) \Phi (z_1) \E J(\xi) \Phi(z_2) = -4 \pa_\xi G_D(\xi,z_1) \pa_\xi G_D(\xi,z_2) . 
\end{align*}
By combining the latter with the residue calculus we obtain 
\begin{align*}
\frac{1}{2\pi i}\int_{\partial D \setminus \R} v_{\zeta} (\xi) \E A(\xi) \Phi(z_1) \Phi(z_2) \,d\xi &=-\frac{2}{\pi i} \int_{\partial D \setminus \R} v_{\zeta} (\xi) \pa_{ \xi }G_D(\xi,z_1) \pa_{\xi} G_D(\xi,z_2) \,d\xi
\\
&=8\pa_{ \zeta }G(\zeta,z_1) \pa_{\zeta} G_D(\zeta,z_2)+2 \Big( v_\zeta(z_1) \pa_{z_1}+ v_\zeta(z_2) \pa_{z_2} \Big) G_D(z_1,z_2). 
\end{align*}
Similarly, we have 
\begin{align*}
 \overline{\frac{1}{2\pi i}\int_{\partial D \setminus \R} v_{\bar{\zeta}} (\xi) \E A(\xi) \Phi(z_1) \Phi(z_2) \,d\xi }
&=-\overline{\frac{2}{\pi i} \int_{\partial D \setminus \R} v_{\bar{\zeta}} (\xi) \pa_{ \xi }G_D(\xi,z_1) \pa_{\xi} G_D(\xi,z_2) \,d\xi }
\\
&= 2 \overline{ \Big( v_{ \bar{\zeta} }(z_1) \pa_{z_1} + v_{ \bar{\zeta} }(z_2) \pa_{z_2} \Big)} G_D(z_1,z_2) 
\\
&= 2 \Big( v_{ \zeta }( \bar{z}_1) \bp_{z_1} + v_{ \zeta }(\bar{z}_2) \bp_{z_2} \Big)G_D(z_1,z_2).
\end{align*}
For $p \in \R,$ let us write 
\begin{equation*}
P(z):=\frac{\pa}{\pa n}\Big|_{\zeta=p} G_D(\zeta,z). 
\end{equation*}
for the Poisson kernel. 
The Hadamard's variational formula for the Loewner chains is given by 
\begin{equation} \label{Hadamard}
\frac{d}{dt}G_{D_t}( g_t(z_1), g_t(z_2) )\Big|_{t=0}=-P(z_1) P(z_2),
\end{equation} 
see \cite{IK13}. Recall that the motion of the moduli under the Loewner flow \cite[Lemma 4.1]{BF08} can be written as 
\begin{equation} \label{motion of slits}
\frac{d}{dt} g_t(z_k^l) = -v_{ \xi_t } ( g_t(z_k^l) ), \qquad \frac{d}{dt} g_t(z_k^r) = -v_{ \xi_t } ( g_t(z_k^r) ).
\end{equation}
Here, $\xi_t= g_t(\gamma(t))$.
Note that the Green's function depends on the underlying domain $D$, which in turn can be interpreted as a function of the moduli $\{ z_k^l, z_k^r, \bar{z}_k^l, \bar{z}_k^r \}$, i.e. 
\begin{equation*}
G_D(\zeta,z)\equiv G_D(\zeta,z; \{ z_k^l, z_k^r, \bar{z}_k^l, \bar{z}_k^r \} ).
\end{equation*}
Denoting 
\begin{equation*}
\pa_k= \pa_{z_k}, \quad (k=1,2), \qquad \pa_k^l= \pa_{ z_k^l }, \quad \pa_k^r = \pa_{ z_k^r }, \quad (k=1,\dots,g),
\end{equation*}
it follows from \eqref{motion of slits} that 
\begin{equation*}
\begin{split}
&\quad -\frac{d}{dt}G_{D_t}( g_t(z_1), g_t(z_2) ) 
= \Big( v_{\xi_t}( g_t(z_1) ) \pa_{1}+ v_{\xi_t}( g_t(z_2) ) \pa_2 \Big) G_{D_t}( g_t(z_1), g_t(z_2) )
\\
&\quad + \sum_{k=1}^g \Big( v_{\xi_t}( g_t(z_k^l) ) \pa_k^l+v_{ \xi_t }( g_t(\bar{z}_k^l) ) \bp_k^l + v_{\xi_t}( g_t(z_k^r) ) \pa_k^r +v_{ \xi_t }( g_t(\bar{z}_k^r) ) \bp_k^r \Big) G_{D_t}( g_t(z_1), g_t(z_2) ) .
\end{split}
\end{equation*}
Since the Green's function is a scalar field (in other words, a $(0,0)$-differential) the Lie derivative formula for scalar fields (see \cite[Proposition 4.1]{KM13}) gives
\begin{equation*}
\begin{split}
\frac{d}{dt}G_{D_t}( g_t(z_1), g_t(z_2) ) \Big|_{t=0} = -\Big( \mathcal{L}_{v_p}+\nabla_{ \mathbb{H}_g } \Big) G_D(z_1,z_2),
\end{split}
\end{equation*}
where $\nabla_{\mathbb{H}_g}$ is given by \eqref{Teich diff op}. Using this and the fact $p\in \R$, the Hadamard's formula \eqref{Hadamard} can be rewritten as 
\begin{align*}
\begin{split}
\Big( \mathcal{L}_{v_p}+\nabla_{ \mathbb{H}_g } \Big) G_D(z_1,z_2)+ 4 \pa_{\zeta} G_D(\zeta,z_1) \pa_\zeta G_D(\zeta,z_2)\Big|_{\zeta=p}=0.
\end{split}
\end{align*} 
The harmonic extension of this identity for $\zeta \in D$ gives rise to 
\begin{equation*}
\Big( \mathcal{L}_{v_\zeta}^+ + \mathcal{L}_{v_{\bar{\zeta}} }^- +\nabla_{ \mathbb{H}_g } \Big) G_D(z_1,z_2)+ 4 \pa_{\zeta} G_D(\zeta,z_1) \pa_\zeta G_D(\zeta,z_2)=0.
\end{equation*}
Combining all of the above, Eguchi-Ooguri equation \eqref{EO 2pt} for $\XX=\Phi(z_1)\Phi(z_2)$ follows.

\subsubsec{Proof that $\Phi(z_1) \cdots \Phi(z_{2n}) \in \mathcal{W}$}
Note that by Wick's formula, we have 
\begin{align*}
\E [A(\xi)\Phi(z_1)\cdots\Phi(z_{2n})] = \sum_{j \neq k}\E [A(\xi)\Phi(z_j)\Phi(z_k) ] \E[\mathcal{X}_{jk}], 
\end{align*}
where 
\[
\mathcal{X}_{jk}=\Phi(z_1)\cdots\Phi(z_{k-1})\Phi(z_{k+1})\cdots\Phi(z_{j-1})\Phi(z_{j+1})\cdots \Phi(z_{2n}), \qquad (j \not = k). 
\]
Then by \eqref{EO eq v1} for the products of two GFFs, it follows that 
\begin{align} \label{eq: recursive_Phi}
\frac{1}{2\pi i} \int_{\partial D \setminus \R} \Big( v_\zeta(\xi)-v_\zeta( \bar{\xi}) \Big) \E A(\xi) \XX \,d\xi = -\sum_{j \neq k} \Big(\nabla_{ \mathbb{H}_g } \E[\Phi(z_j)\Phi(z_k)]\Big) \E[\mathcal{X}_{jk}]. 
\end{align}
On the other hand, by Wick's formula, we have 
\begin{equation} \label{n-pt cor}
\E [ \Phi(z_1) \cdots \Phi(z_{2n}) ] = \sum_{\mathcal{P}} \prod_k \E [ \Phi(z_{i_k})\Phi(z_{j_k}) ],
\end{equation}
where the sum is taken over all partitions $\mathcal{P}$ of the set $\{ 1,\cdots ,2n \}$ into disjoint pairs $\{ i_l ,j_l\}$.
Then by Leibniz's rule and rearranging the terms, we have
	\begin{align*}
\nabla_{ \mathbb{H}_g } \E [ \Phi(z_1) \cdots \Phi(z_{2n}) ] & = \sum_{\mathcal{P}} \Big( \prod_k \E [ \Phi(z_{i_k})\Phi(z_{j_k}) ] \Big) \Big( \sum_k \frac{ \nabla_{ \mathbb{H}_g } \E [ \Phi(z_{i_k})\Phi(z_{j_k}) ]}{\E [ \Phi(z_{i_k})\Phi(z_{j_k}) ]}\Big)
\\
&= \sum_{j \neq k} \Big( \nabla_{ \mathbb{H}_g } \E[\Phi(z_j)\Phi(z_k)]\Big) \Big( \sum_{\mathcal{P}_{jk}} \prod_l \E [ \Phi(z_{i_l})\Phi(z_{j_l}) ] \Big), 
\end{align*} 
where $\mathcal{P}_{jk}$ is the set of all partitions of $\{1, \ldots, 2n \} \setminus \{ j,k \}$ into disjoint pairs. By another application of \eqref{n-pt cor} the inner sum over $\mathcal{P}_{jk}$ is simply $\E \mathcal{X}_{jk}$. Combining this with \eqref{eq: recursive_Phi} we obtain \eqref{EO eq v1} for $\mathcal X=\Phi(z_1)\cdots\Phi(z_{2n})$. 
	
\subsubsec{$\mathcal{W}$ is closed under OPE products}
Suppose that we have the following OPE
\[
X(\zeta)Y(z ) = \sum_n C_n(z)(\zeta -z)^n,\qquad \zeta \to z.
\] 
If both $X$ and $Y$ are in $\mathcal W$, as $\zeta \to z$, 
\begin{align*}
\frac{1}{2\pi i} \int_{\partial D \setminus \R} \left( v_\zeta(\xi)-v_\zeta( \bar{\xi}) \right) \E A(\xi) X(\zeta)Y(z) \,d\xi & = \sum_n \left( \frac{1}{2\pi i} \int_{\partial D \setminus \R} \left( v_\zeta(\xi)-v_\zeta( \bar{\xi}) \right) \E A(\xi) C_n(z) \, d\xi \right) (\zeta-z)^n 
\end{align*}
and 
\[
\nabla_{ \mathbb{H}_g } \E[X(\zeta)Y(z)] = \sum_n \nabla_{ \mathbb{H}_g } \E[C_n(z)] (\zeta-z)^n.
\] 
By comparing the coefficients of the two expansions above, it follows that $C_n \in \mathcal W$.
\end{proof}

\begin{rem}
By Lemma~\ref{Lem_EO beta0} with $\XX=\Phi(z_1) \Phi(z_2)$, we have
\begin{equation*}
\frac{1}{\pi i} \int_{\partial D \setminus \R} \Big( v_\zeta(\xi)-v_\zeta( \bar{\xi}) \Big)\pa_{ \xi }G_D(z_1,\xi) \pa_{\xi} G_D(\xi,z_2) \,d\xi = \nabla_{ \mathbb{H}_g } G_D(z_1,z_2).
\end{equation*}
This gives 
\begin{equation} \label{EO G pm}
\frac{1}{\pi i} \int_{\partial D \setminus \R} \Big( v_\zeta(\xi)-v_\zeta( \bar{\xi}) \Big)\pa_{ \xi }G^\pm_D(z_1,\xi) \pa_{\xi} G_D(\xi,z_2) \,d\xi = \nabla_{ \mathbb{H}_g } G^\pm_D(z_1,z_2).
\end{equation}
The last formula will be used in the proof of the next proposition.
\end{rem}

We now extend the generalised Eguchi-Ooguri equation to a non-trivial background charge $\bfs \beta.$

\begin{prop} 
For any $\XX \in \FF_{\bfs \beta}$, the generalized Eguchi-Ooguri equation
\begin{equation} \label{eq: EG_generalized_beta}
\frac{1}{2\pi i} \int_{ \pa D   \setminus \R   } \Big( v_\zeta(\xi)-v_\zeta( \bar{\xi}) \Big) \E A_{ \bfs \beta}(\xi) \XX \,d\xi +\nabla_{ \mathbb{H}_g } \E \XX=0
\end{equation}
holds in the identity chart of a chordal standard domain $D$. 
\end{prop}

\begin{proof}
The proof proceeds in several steps, as in the proof of Lemma \ref{Lem_EO beta0}.

\subsubsec{Identity \eqref{eq: EG_generalized_beta} holds for $\Phi_{\bfs \beta}$}
We show the identity for 
\begin{equation*}
\XX=\Phi_{\bfs{\beta}}= \Phi+\varphi_{ \bfs{\beta} }, \qquad \varphi_{ \bfs{\beta} }(z)= i \sum \beta_k \Big( G_D^+(q_k,z)-G_D^-(q_k,z) \Big).
\end{equation*} 
Recall that 
\begin{equation*}
A_{ \bfs{\beta} }:= A_{ (0) }+\Big(ib\pa-\jmath_{ \bfs{\beta} } \Big)J.
\end{equation*}
Applying Wick's formula, we have 
\begin{align*}
\begin{split}
\E A_{\bfs{\beta}}(\xi) \Phi_{ \bfs{\beta} }(z)
&= \E \Big[ \Big(ib\pa_\xi-\jmath_{ \bfs{\beta} }(\xi) \Big)J(\xi) \Big( \Phi(z)+\varphi_{ \bfs{\beta} }(z) \Big) \Big] 
\\
&= ib\pa_\xi \E \Big[ J(\xi) \Big( \Phi(z)+\varphi_{ \bfs{\beta} }(z) \Big) \Big] -\jmath_{ \bfs{\beta} }(\xi)\,\E J(\xi) \Phi_{ (0) }(z) 
\\
&= \Big( ib\pa_\xi -\jmath_{ \bfs{\beta} }(\xi) \Big)\,\E J(\xi) \Phi_{ (0) }(z) .
\end{split}
\end{align*}

\subsubsec{Reduction to an integral identity} Since $\E\,\Phi_{\bfs \beta} = \varphi_{\bfs \beta}$, the last computation means that we want to show
\begin{align}\label{eq:Phib_Ward_proof_identity}
\frac{1}{2\pi i} \int_{\partial D \setminus \R} \left( v_\zeta(\xi)-v_\zeta( \bar{\xi}) \right) (ib \pa_{\xi} - j_{\bfs \beta}(\xi)) \E J(\xi) \Phi_{(0)}(z) \, d\xi + \nabla_{ \mathbb{H}_g } \varphi_{\bfs \beta} = 0.
\end{align}
However by means of the asymptotic $\zeta \to z$ expansions
\begin{equation*}
v_\zeta(z)=\frac{2}{\zeta-z}+O(1), \qquad \pa_{ \zeta }^2 G(\zeta,z)=\frac12 \frac{1}{(\zeta-z)^2}+O(1), 
\end{equation*}
and the residue calculus, we have
\begin{align*}
\frac{1}{2\pi i} \int_{\partial D \setminus \R} v_\zeta(\xi) \pa_\xi^2 G_D(\xi,z)\,d\xi &= -2 \pa_\zeta^2 G_D(\zeta,z)+ \frac12 v_\zeta'(z) 
\end{align*}
and
\begin{align*}
 \overline{\frac{1}{2\pi i}\int_{\partial D \setminus \R} v_{\bar{\zeta}} (\xi) \pa_\xi^2 G_D(\xi,z) \,d\xi }
&= -\frac12 \overline{ v_{ \bar{\zeta} }'(z) }. 
\end{align*}
This gives 
\begin{align*}
\frac{1}{2\pi i} \int_{\partial D \setminus \R} \Big( v_\zeta(\xi)-v_\zeta( \bar{\xi}) \Big) \pa_\xi \E J(\xi) \Phi_{ (0) }(z) \,d\xi &= -4 \pa_\zeta^2 G_D(\zeta,z)+ v_\zeta'(z) - \overline{ v_{ \bar{\zeta} }'(z) } = 0, 
\end{align*}
with the last equality following from Lemma~\ref{lem: Green_crucial_identity}. The above implies that showing \eqref{eq:Phib_Ward_proof_identity} is equivalent to showing
\begin{align}\label{eq: Phibeta_integral_identity}
\frac{1}{2\pi i} \int_{\partial D \setminus \R} \left( v_\zeta(\xi)-v_\zeta( \bar{\xi}) \right) j_{\bfs \beta}(\xi) \E J(\xi) \Phi_{(0)}(z) \, d\xi = \nabla_{ \mathbb{H}_g } \varphi_{\bfs \beta}.
\end{align}

\subsubsec{Proof of \eqref{eq: Phibeta_integral_identity}}
This identity is relatively straightforward. We have that
\begin{align*}
&\quad \frac{1}{2\pi i} \int_{\partial D \setminus \R} \Big( v_\zeta(\xi)-v_\zeta( \bar{\xi}) \Big) \jmath_{ \bfs{\beta} }(\xi)\,\E J(\xi) \Phi_{ (0) }(z) \,d\xi 
\\
&= \frac{1}{ \pi i} \int_{\partial D \setminus \R} \Big( v_\zeta(\xi)-v_\zeta( \bar{\xi}) \Big) \jmath_{ \bfs{\beta} }(\xi)\, \partial_\xi G_D(\xi,z) \,d\xi
\\
&= \frac{2}{\pi } \sum \beta_k \int_{\partial D \setminus \R} \Big( v_\zeta(\xi)-v_\zeta( \bar{\xi}) \Big) \pa_\xi \Big(G_D^+(q_k,\xi)-G_D^-(q_k,\xi)\Big) \partial_\xi G^+(\xi,z) \,d\xi \\
&=2 i \sum \beta_k \nabla_{ \mathbb{H}_g } \Big( G_D^+(q_k,z)-G_D^-(q_k,z) \Big)=\nabla_{ \mathbb{H}_g } \varphi_{ \bfs{\beta} }(z).
\end{align*}
The second equality uses the definition of $j_{\bfs \beta}$, while the third uses \eqref{EO G pm}. This proves \eqref{eq:Phib_Ward_proof_identity} and \eqref{eq: Phibeta_integral_identity}.

\subsubsec{Identity \eqref{eq: EG_generalized_beta} holds for $\XX = \Phi_{\bfs \beta}(z_1) \Phi_{\bfs \beta}(z_2)$}
We now prove the identity \eqref{eq: EG_generalized_beta} for $\XX = \Phi_{ \bfs{\beta} }(z_1) \Phi_{ \bfs{\beta} }(z_2)$. By expanding out $A_{\bfs \beta}$ and $\varphi_{\bfs \beta}$ we obtain
\begin{align*}
\frac{1}{2\pi i} \int_{\partial D \setminus \R} &\Big( v_\zeta(\xi)-v_\zeta( \bar{\xi}) \Big) \E A_{\bfs{\beta}}(\xi) \Phi_{ \bfs{\beta} }(z_1)  \Phi_{ \bfs{\beta} }(z_2) \,d\xi\\
&= \frac{1}{2\pi i} \int_{\partial D \setminus \R} \Big( v_\zeta(\xi)-v_\zeta( \bar{\xi}) \Big) \E A_{ (0) }(\xi)\Phi_{ (0) }(z_1)\Phi_{ (0) }(z_2) \,d\xi \\
&- \frac{ \varphi_{ \bfs{\beta} }(z_2) }{2\pi i} \int_{\partial D \setminus \R} \Big( v_\zeta(\xi)-v_\zeta( \bar{\xi}) \Big) \left( ib \partial_{\xi} - \jmath_{ \bfs{\beta} }(\xi) \right) \,\E J(\xi) \Phi_{ (0) }(z_1) \, d\xi \\
& - \frac{ \varphi_{ \bfs{\beta} }(z_1) }{2\pi i} \int_{\partial D \setminus \R} \Big( v_\zeta(\xi)-v_\zeta( \bar{\xi}) \Big) \left( ib \partial_{\xi} - \jmath_{ \bfs{\beta} }(\xi) \right) \,\E J(\xi) \Phi_{ (0) }(z_2) \, d\xi \\
&= \nabla_{ \mathbb{H}_g } \E\,\Phi_{ (0) }(z_1)\Phi_{ (0) }(z_2) + \varphi_{ \bfs{\beta} }(z_2)\nabla_{ \mathbb{H}_g }\varphi_{ \bfs{\beta} }(z_1) + \varphi_{ \bfs{\beta} }(z_1)\nabla_{\mathbb{H}_g} \varphi_{ \bfs{\beta} }(z_2) \\
&= \nabla_{ \mathbb{H}_g } \E\,\Phi_{ \bfs{\beta} }(z_1)\Phi_{ \bfs{\beta} }(z_2).
\end{align*}
The second equality is a combination of Lemma \ref{Lem_EO beta0} and identity \eqref{eq:Phib_Ward_proof_identity}, and the third equality is a Leibniz formula. 

\subsubsec{Identity \eqref{eq: EG_generalized_beta} holds for $\XX = \Phi_{\bfs \beta}(z_1) \cdots \Phi_{\bfs \beta}(z_{n})$}
To handle the general case $n>2$, for each $k \le n$, let 
\[
I_k:=\sum \Big(\prod_{l=1}^{k}\varphi_{ \bfs{\beta} }(z_{i_l}) \Big) \Big(\prod_{l=1}^{n-k} \Phi_{ (0) }(z_{j_l}) \Big),
\]
where the sum is over all partitions of $\{ 1, \ldots, n \}$ into a set of set $k$ and the complementary set of size $n-k$. Then we have 
\[
\Phi_{ \bfs{\beta} }(z_1)\cdots \Phi_{ \bfs{\beta} }(z_n) =I_0+\cdots+I_n.
\]
Let us write 
\begin{align*}
A_k&:=\sum \prod_{l=1}^{k}\varphi_{ \bfs{\beta} }(z_{i_l}) \, \nabla_{ \mathbb{H}_g } \E \prod_{l=1}^{n-k} \Phi_{ (0) }(z_{j_l}),
\\ 
B_k &:=\sum_{j=1}^{n} \nabla_{ \mathbb{H}_g } \varphi_{ \bfs{\beta} }(z_j) \sum \prod_{l=1}^{k-1}\varphi_{ \bfs{\beta} }(z_{i_l}) \E \prod_{l=1}^{n-k} \Phi_{ (0) }(z_{j_l}).
\end{align*}
Then it follows from the Leibniz rule that $ \nabla_{ \mathbb{H}_g } \E I_k = A_k + B_k$. On the other hand, we have 
\[
\frac{1}{2\pi i} \int_{\partial D \setminus \R} \Big( v_\zeta(\xi)-v_\zeta( \bar{\xi}) \Big) \E A_{ \bfs \beta}(\xi) I_k \,d\xi = C_k+D_k,
\]
where 
\begin{align*}
	C_k&:= 	\frac{1}{2\pi i} \int_{\partial D \setminus \R} \Big( v_\zeta(\xi) - v_\zeta( \bar{\xi}) \Big) \E A(\xi) I_k \, d\xi,
		\\ 
		D_k&:= -	\frac{1}{2\pi i} \int_{\partial D \setminus \R} \Big( v_\zeta(\xi)-v_\zeta( \bar{\xi}) \Big)  \E (ib \partial J(\xi) - j_{\bfs \beta}(\xi) J(\xi)) I_k \, d\xi.
\end{align*}
Lemma \ref{Lem_EO beta0} gives that $A_k=C_k$, and an application of Wick's formula and \eqref{eq:Phib_Ward_proof_identity} gives that $B_{k}=D_k$. This completes the proof. 
\end{proof}

We finish this section with a proof of the crucial Green's function identity \eqref{eq: Greens_function_crucial_identity}. Our proof is based on properties of the bipolar Green's function $G_{p,q}$ on the Schottky double.

\begin{lem}\label{lem: Green_crucial_identity}
Let $D$ be a chordal standard domain and $G = G^{Diri}$ be the Dirichlet Green's function for $D$. Then in the identity chart of $D$ we have
\begin{equation*}
4 \pa_{\zeta}^2 G(\zeta, z) = v_{\zeta}'(z) - \overline{v_{\zeta}'(z)},
\end{equation*} 
where $v_{\zeta}(z) = -4 \pa_{\zeta} G^{ER,+}(z, \zeta)$.
\end{lem}

\begin{proof}
For $\zeta \in \mathbb{R}$ we have defined the vector field $v_{\zeta}$ by \eqref{def: v_zeta}, which is equivalent to
\[
v_{\xi}(z) = -4 \left. \pa_{\zeta} \right|_{\zeta = \xi} G^{ER, +}(z, \zeta).
\]
With this representation of $v_{\zeta}$ in hand, proving the lemma is equivalent to showing
\begin{align}\label{eq:GF_necessary_identity}
\pa_{\zeta}^2 G(\zeta, z) = - \pa_{\zeta} \left( \pa_z G^{ER, +}(z, \zeta) - \overline{\pa_z G^{ER, +}(z, \overline{\zeta})} \right)
\end{align}
By means of the bipolar Green's function on the Schottky double we have
\[
G(\zeta, z) = G^{Diri}(\zeta, z) = G_{z, z^*}(\zeta) - G_{z, z^*}(\zeta^*), 
\]
from which it follows that
\[
\pa_{\zeta}^2 G(\zeta, z) = \pa_{\zeta}^2 G_{z, z^*}(\zeta).
\]
By representation \eqref{ER Green v2} of the ER Green's function we have
\[
-\pa_{\zeta} \pa_z G^{ER, +}(z, \zeta) = \tfrac{1}{2} \pa_{\zeta} \pa_z \log \Theta(\mathcal A(z) - \mathcal A(\zeta) - e) =: \Omega_z(\zeta),
\]
so that 
\[
- \pa_{\zeta} \left( \pa_z G^{ER, +}(z, \zeta) - \overline{\pa_z G^{ER, +}(z, \overline{\zeta})} \right) = \Omega_z(\zeta) - \Omega_{z^*}(\zeta).
\]
Let $\omega_1 = \pa_{\zeta}^2 G_{z, z^*}(\zeta) \, d \zeta$ and $\omega_2 = (\Omega_z(\zeta) - \Omega_{z^*}(\zeta)) \, d \zeta$. We will prove \eqref{eq:GF_necessary_identity} by showing that $\omega_0 := \omega_1 - \omega_2 = 0$ on the Schottky double. As a first step, note that both $\omega_1$ and $\omega_2$ have double poles at $z, z^*$, and no poles elsewhere, and that consequently $\omega_0$ has no singularities on the Schottky double. Hence it is a holomorphic differential. Moreover we have that
\[
\oint_{a_j} \omega_1 = \oint_{a_j} d \left( \pa_{\zeta} G_{z, z^*}(\zeta) \right) = 0
\]
for every cycle $a_j$, and by definition of $\Omega_z(\zeta)$ we have
\[
\oint_{a_j} \omega_2 = \pa_{\eta} |_{\eta = z} (\log \theta(\zeta + a_j - \eta) - \log \theta(\zeta - \eta)) - \pa_{\eta} |_{\eta = z^*} (\log \theta(\zeta + a_j - \eta) - \log \theta(\zeta - \eta)) = 0,
\]
where $\theta(z) = \Theta(\mathcal{A}(z) - e)$ for short. The right hand side is zero by periodicity of $\theta$ around the $a_j$ cycles. Thus $\omega_0 = \omega_1 - \omega_2$ is a holomorphic differential with vanishing $a_j$ periods. Consequently $\omega_0 \equiv 0$, by a well known theorem in complex analysis, see \cite[Proposition III.3.3]{FK92}.
\end{proof}

\begin{rem}
The quantity $\Omega_z(\zeta) \, d \zeta \, dz$ is often known as the fundamental normalized bidifferential. See \cite[Chapter 2]{Fay73} and \cite[Section 2.3]{Du15} for more. It also makes an appearance in the recent work \cite{IK13}, which studies scaling limits of Ising model correlation functions in multiply connected domains. 
\end{rem}

\subsection{BPZ equations} \label{Subsection_BPZ}

In this section we prove the BPZ equations for $\FF_{\bfs \beta}$ on a chordal standard domain. The BPZ equations may be regarded as a limiting form of the Ward equation as the argument of the Virasoro field $T_{\bfs \beta}$ approaches the node of a field $X \in \FF_{\bfs \beta}$. The limit is in the sense of the operator product expansion between the Virasoro field and $X$. Operator product expansion with $T_{\bfs \beta}$ is expressed in terms of the Virasoro generators $L_n \equiv L_n^{\bfs \beta}$. Recall that for $n \in \mathbb Z$, the Virasoro generator $L_n \equiv L_{n}^{ \bfs \beta }$ is an operator that acts on fields $X \in \FF_{\bfs \beta}$ via 
\begin{equation} \label{Virasoro generator}
L_n^{ \bfs \beta }(z):=\frac{1}{2\pi i}\oint_{(z)}(\zeta -z )^{n+1} T_{ \bfs{\beta} }(\zeta) \, d\zeta, \qquad \textup{i.e.} \qquad L_n^{ \bfs \beta } X= T_{ \bfs \beta } \ast_{ (-n-2) } X. 
\end{equation}
Equivalently, the OPE expansion between $T_{\bfs \beta}$ and $X$ can be expressed in terms of the Virasoro nodes as 
\begin{equation} \label{T Ln}
T_{ \bfs \beta } X(z) = \cdots + \frac{ (L_0^{ \bfs \beta } X)(z) }{ (\zeta-z)^2 } + \frac{ (L_{-1}^{ \bfs \beta }X)(z) }{ \zeta-z } + (L_{-2}^{ \bfs \beta } X)(z) +\cdots. 
\end{equation}
To express the BPZ equations we also introduce the vector field $k_{\zeta}$ which has a single pole at $\zeta$, defined as
\begin{equation*}
k_{\zeta}(z)= \frac{2}{\zeta-z}.
\end{equation*}

\begin{prop}\label{Prop_BPZ T}
For $Y\in \FF_{\bfs\beta}$ and $X=X_1\cdots X_n$, $( X_j \in \FF_{\bfs\beta}$), we have
\begin{equation}
\begin{split} \label{BPZ Y gen}
2\,\E\,T_{ \bfs \beta } \ast Y(z) X& = \E\,Y(z) \LL_{v_z}^+ X+ \LL_{v_{\bar{z}}}^-\E\,Y(z)X +  \nabla_{ \mathbb{H}_g , \bfs q }\,   \E\,Y(z)X
\\
&\quad + \lim_{\zeta \to z} \Big( \LL_{v_\zeta}^+- \LL_{k_\zeta}^+ \Big) \E\,Y(z) X+2\,\E\,T_{\bfs\beta} (z)\, \E\,Y(z) X , 
\end{split}
\end{equation}
where all fields are evaluated in the identity chart of the chordal standard domain. 
In particular, for a holomorphic differential $Y\in \FF_{\bfs\beta}$ with conformal dimension $h$, we have 
\begin{equation}
\begin{split} \label{BPZ Y differential}
2\,\E\,T_{ \bfs \beta } \ast Y(z) X& = \E\,Y(z) \LL_{v_z}^+ X+ \LL_{v_{\bar{z}}}^-\E\,Y(z)X +   \nabla_{ \mathbb{H}_g , \bfs q } \,   \E\,Y(z)X
\\
&\quad + \Big[ h \Big(r_{D,0}'(z)-r_{D,1}(z)\Big) + 2\,\E\,T_{\bfs\beta}(z) + r_{D,0}(z) \pa_z \Big] \E\,Y(z) X, 
\end{split}
\end{equation}
where $r_{D,0}$ and $r_{D,1}$ are given by \eqref{v diagonal}. 
\end{prop}
\begin{proof}
It follows from Ward's OPE \cite[Proposition 5.3]{KM13} that
\begin{equation*}
\Sing_{\zeta \to z}A_{\bfs\beta}(\zeta)Y(z)X= \frac12\Big( \LL_{ k_\zeta }^+ Y(z) \Big) X.
\end{equation*}
On the other hand, by Theorem~\ref{Thm_Ward},
\begin{equation}
\begin{split} \label{A YX BPZ}
2 \, \E A_{ \bfs \beta }(\zeta) Y(z)X & = \Big( \LL_{v_\zeta}^+ + \LL_{ v_{\bar{\zeta}} }^- \Big) \E\,Y(z)X +   \nabla_{ \mathbb{H}_g , \bfs q } \,  \E\,Y(z)X.
\end{split}
\end{equation}
We now subtract the singular part from the both sides of \eqref{A YX BPZ} when $\zeta \to z$.
Consequently, we obtain
\begin{equation*}
2\,\E A_{\bfs\beta}\ast Y(z) X= \E\,Y(z) \LL_{v_z}^+ X+ \LL_{v_{\bar{z}}}^-\E\,Y(z)X +   \nabla_{ \mathbb{H}_g , \bfs q } \,  \E\,Y(z)X + \lim_{\zeta \to z} \Big( \LL_{v_\zeta}^+- \LL_{k_\zeta}^+ \Big) \E\,Y(z) X. 
\end{equation*}
Then the identity \eqref{BPZ Y gen} follows from $T_{ \bfs \beta }= A_{ \bfs \beta }+ \E\,T_{ \bfs \beta } $. 

Recall that if the field $Y$ is a differential, we have 
\begin{align*}
 \LL_{v_\zeta}^+ Y(z) = \Big( hv_{\zeta}'(z)+ v_{\zeta}(z) \pa_z \Big) Y(z), \qquad 
 \LL_{k_\zeta}^+ Y(z) = \Big( \frac{2h}{(\zeta-z)^2}+ \frac{2}{\zeta-z} \pa_z \Big) Y(z),
\end{align*}
see \cite[Proposition 4.1]{KM13}.
It now follows from \eqref{v diagonal} and \eqref{v' diagonal} that 
\begin{align*}
\lim_{\zeta \to z} \Big( \LL_{v_\zeta}^+ - \LL_{k_\zeta}^+ \Big) Y(z) & = \lim_{\zeta \to z} \Big[ h \Big( v_\zeta'(z)-\frac{2}{(\zeta-z)^2} \Big) + \Big( v_\zeta(z)-\frac{2}{\zeta-z} \Big) \pa_z \Big]Y(z)
\\
&=\Big[ h \Big(r_{D,0}'(z)-r_{D,1}(z)\Big) + r_{D,0}(z) \pa_z \Big] Y(z). 
\end{align*}
This completes the proof.
\end{proof}

Proposition \ref{Prop_BPZ T} is especially fruitful when the field $Y$ is \textit{Virasoro primary}. By definition, a Virasoro primary field $X \in \FF_{\bfs{\beta}}$ is a differential of conformal dimension $(\lambda,\lambda_*)$. 
Then by \cite[Proposition 7.5]{KM13}, it follows that 
\begin{equation} \label{L0 L-1 Ln}
L_0^{ \bfs \beta } X=\lambda X, \qquad L_{-1}^{ \bfs \beta } X=\pa X, \qquad L_n^{ \bfs \beta } X=0 \quad \textup{for all } n \ge 1. 
\end{equation}
Similar identities also hold for $\bar{X}$. 
Furthermore, $X$ is called \emph{current primary} if there exist $q, q_*$ such that
\begin{equation*}
J_0^{ \bfs \beta } X=-iq X, \qquad J_0^{ \bfs \beta } \bar{X}=i\overline{q}_* \bar{X}, \qquad J_n^{ \bfs \beta } X=J_n^{ \bfs \beta } \bar{X} =0 \quad \textup{for all } n \ge 1, 
\end{equation*}
where the numbers $q,q^*$ are called charges of $X$ and the current generators are given by 
\begin{equation*}
J_n^{ \bfs \beta }(z):=\frac{1}{2\pi i}\oint_{(z)} (\zeta -z )^n J_{ \bfs{\beta} }(\zeta) \, d\zeta, \qquad \textup{i.e.} \qquad J_n^{ \bfs \beta } X= J_{ \bfs \beta } \ast_{ (-n-1) } X. 
\end{equation*}

\begin{prop}[Cf. Proposition 11.2 \cite{KM13}] \label{prop: Level_Two}
Let $\OO \in \FF_{ \bfs \beta }$ be a current primary field with charges $q$, $q_*$. 
Then if $2q(b+q)=1$, we have
	\begin{equation} \label{level2}
		L_{-2}^{ \bfs \beta } \OO = \frac{1}{2q^2}(L^{\bfs \beta}_{-1})^2 \OO.
	\end{equation}
\end{prop}

By combining Proposition \ref{prop: Level_Two} with Proposition \ref{Prop_BPZ T} we can now give the proof of Theorem \ref{Thm_BPZ}, the BPZ-Cardy equations for a chordal standard domain.

\begin{proof}[Proof of Theorem~\ref{Thm_BPZ}]
Recall that the field $\Upsilon(z)$ is defined by
\[
\Upsilon(z) := \OO \{ \bfs \alpha \}, \qquad \bfs \alpha= ( a \cdot z - a \cdot \infty, \bfs 0 ), 
\]
where $a \in \R$ satisfies $2a(a+b)=1$ and $b \in \R$ is determined by the neutrality condition \eqref{NC_b} of $\bfs \beta$. Then $\Upsilon(z)$ is primary with dimensions $\lambda = a^2/2 - ab = h_{1,2}$ and current primary with charge $a$, see \cite[Proposition 12.4]{KM13}. Thus by combining the level two equation \eqref{level2} with \eqref{Virasoro generator} and \eqref{L0 L-1 Ln}, we obtain
\begin{equation*}
T_{ \bfs \beta } \ast \Upsilon (z) =\frac{1}{2a^2} \pa_z^2 \Upsilon(z) . 
\end{equation*}
Then the BPZ equation \eqref{BPZ eq} follows from Proposition~\ref{Prop_BPZ T}. 

Now to derive the BPZ-Cardy equation \eqref{BPZ-Cardy}, first rewrite the BPZ equation \eqref{BPZ eq} in terms of $ \E_\Upsilon X$ via
\begin{align*}
\frac{1}{a^2} \pa_z^2 &\Big( \E_\Upsilon X(z) \E \Upsilon(z) \Big) = \frac{1}{a^2} \Big[ \E \Upsilon(z) \pa_z^2 \E_\Upsilon X(z) + 2 \pa_z \E_\Upsilon X(z) \pa_z \E \Upsilon(z) +\E_\Upsilon X(z) \pa_z^2 \E \Upsilon(z) \Big]
\\
& = \Big[ \check{\LL}_{v_z}^+ + \LL_{v_{\bar{z}}}^- +    \nabla_{ \mathbb{H}_g , \bfs q }  + h_{1,2} \Big(r_{D,0}'(z)-r_{D,1}(z)\Big) + 2\,\E\,T_{\bfs\beta}(z) + r_{D,0}(z) \pa_z \Big] \E_\Upsilon X(z) \E \Upsilon(z).
\end{align*}
On the other hand, by multiplying $\E_\Upsilon X $ to the null-vector equation \eqref{null vector eq}, we have
\begin{align*}
\frac{1}{a^2} \E_\Upsilon X \pa_z^2 \E \Upsilon(z) & = 
 \E_\Upsilon X \Big[ \nabla_{ \mathbb{H}_g , \bfs q }   + h_{1,2} \Big(r_{D,0}'(z)-r_{D,1}(z)\Big) + 2\,\E\,T_{\bfs\beta}(z) + r_{D,0}(z) \pa_z \Big] \E \Upsilon(z) . 
\end{align*}
Subtracting these equations gives rise to 
\begin{align*}
\begin{split} 
\frac{1}{a^2} &\Big( \E \Upsilon(z) \pa_z^2 \E_\Upsilon X(z) + 2 \pa_z \E_\Upsilon X(z) \pa_z \E \Upsilon(z) \Big)
\\
& =\E \Upsilon(z) \Big( \check{\LL}_{v_z}^+ + \LL_{v_{\bar{z}}}^- +  \nabla_{ \mathbb{H}_g , \bfs q }   + r_{D,0}(z) \pa_z \Big) \E_\Upsilon X(z) .
\end{split}
\end{align*}
Therefore we conclude \eqref{BPZ-Cardy}.
\end{proof}

\section{SLE Martingale Observables \label{sec: SLE_MO}}

In this section we prove Theorem \ref{Thm_SLE MO}, that the correlation functions of all fields in the family $\FF_{\bfs \beta}$ are martingale observables for an associated SLE process. The main ingredient is the BPZ-Cardy equation \eqref{BPZ-Cardy}, which was derived without any reference to SLE. Theorem \ref{Thm_SLE MO} gives an infinite family of martingale observables for the SLE process, and the martingale can depend on any finite collection of marked points in the domain. Note that the BPZ-Cardy equation allows us to verify the martingale observable property without having explicit formulas for the martingale. In fact we typically do not even have an explicit formula for the Loewner vector field $v_{\xi}$ or the stochastic differential equation \eqref{driving SDE}. Nonetheless the martingale observable property is verified by a relatively straightforward application of Ito's formula, thereby demonstrating the power of the CFT approach.

\begin{proof}[Proof of Theorem~\ref{Thm_SLE MO}]
Using the BPZ-Cardy equation \eqref{BPZ-Cardy}, the proof follows along the same lines in the literature \cite[Section 14.3]{KM13}, \cite[Section 8.3]{KM21} and \cite[Section 5.3]{BKT23}. We give a shortened version here.

Recall that $M$ is the non-random conformal field
\[
M = \E_{\Upsilon} \XX = \frac{\E \Upsilon(\xi) \XX}{\E \Upsilon(\xi) },
\]
where $\Upsilon(\xi) = \OO_{\bfs \beta}\{ a \cdot \xi - a \cdot \infty \}$. Thus $M$ depends on the marked point $\xi$, the marked points in $\bfs \beta$, the nodes of $\XX$, and the modular parameters $z_k^l$, $z_k^r$ of the chordal standard domain. Denote by $\bfs z$ the modular parameters and $\bfs q$ the marked points of $\bfs \beta$. Define a function $m(\xi, \bfs z, \bfs q, t)$ by 
\[
m(\xi, \bfs z, \bfs q, t) :=(R_{\xi, \bfs z, \bfs q} \, \| \, g_t^{-1}), \qquad R_{\xi, \bfs z, \bfs q} = \E_{\Upsilon} \XX.
\]
We need to prove that $M_t := m(\xi_t, \bfs z(t), \bfs q(t), t)$ is a martingale, where $\xi_t$ is the solution to the stochastic differential equation defined by \eqref{driving SDE} and \eqref{drift function}. By Ito's formula and \eqref{drift function} we have
\begin{align*}
d M_t &= (\pa_{\xi} m)(\xi_t, \bfs z(t), \bfs q(t), t) \, d \xi_t + \frac{1}{2} (\pa_{\xi}^2 m)(\xi_t, \bfs z(t), \bfs q(t), t) \, d \langle \xi \rangle_t + \frac{\pa m}{\pa t}(\xi_t, \bfs q(t),t) \, dt 
\\
&\quad + \sum_k \left(\dot z_k(t) \pa_{z_k}m(\xi_t, \bfs z(t), \bfs q(t), t) + \overline{\dot z_k(t)} \bp_{z_k}m(\xi_t, \bfs z(t), \bfs q(t), t) \right) \, dt
\\
&\quad + \sum_k \dot q_k(t) \pa_{q_k} m(\xi_t, \bfs z(t), \bfs q(t), t) \\
&= \sqrt{\kappa} (\pa_{\xi} m) \, dB_t + \left( \Big(\kappa \, \pa_{\xi} \log \E \Upsilon(\xi) - r_0(\xi) \Big) \pa_{\xi} m + \frac{\kappa}{2} \pa_{\xi}^2 m + \frac{\pa m}{\pa t} -\nabla_{ \mathbb{H}_g, \bfs q } m \right) dt. 
\end{align*}
The drift term vanishes by the BPZ-Cardy equations, we will now explain why. By definition of $m$ the time derivative $\pa m/\pa t$ is given by
\[
\frac{\pa m}{\pa t}(\xi, \bfs z, \bfs q, t) = \left. \frac{d}{ds} \right|_{s=0} \E \XX \, \| \, g_{t+s}^{-1}.
\]
Let $f_{s,t} = g_{t+s} \circ g_t^{-1}$. Then $\pa_s f_{s,t} = -v_{\xi_{t+s}} \circ f_{s,t}$ by the Loewner equation, with $f_{0,t} = \id$. We deduce that
\[
f_{s,t} = \id - s v_{\xi_t} + o(s), \quad s \downarrow 0.
\]
Since $g_{t+s} = f_{s,t} \circ g_t$, the definition of the Lie derivative gives
\[
\left. \frac{d}{ds} \right|_{s=0} \E \XX \, \| \, g_{t+s}^{-1}  = \left. \frac{d}{ds} \right|_{s=0} \E \XX \, \| \, g_{t}^{-1} \circ f_{s,t}^{-1}  = -\check{\LL}_{v_{\xi}} m \, \| \, g_t^{-1}.
\]
Taken together and invoking the BPZ-Cardy equations \eqref{BPZ-Cardy} (with $\xi$ on the real line), the last three calculations show that the drift term of $d M_t$ vanishes.
\end{proof}

We conclude with some example of martingale observables. 
We first discuss the bosonic observable for these SLE systems, i.e. the one-point function of the modified Gaussian field $\Phi_{\bfs \beta}$. Here the modification comes in two different forms: one in the background charge $\bfs \beta$ and another through the OPE exponential $\Upsilon(\xi)$. This modification affects both the structure of the martingale observable and the associated SLE system. It is the analogue of the Schramm-Sheffield observable for chordal SLE$(\kappa, \bfs \rho)$ systems in simply connected domains.

\begin{example}{Bosonic observable}
The simplest one-point observable is obtained by taking $X = \Phi_{\bfs \beta}(z)$, typically for $z$ in the interior of the chordal standard domain. Recall that $\bfs \beta = \sum_k \beta_k \cdot q_k$ for $q_k \in \R$ and the $\beta_k$ satisfy the neutrality condition $\sum \beta_k = b \chi$. Here, we have 
$
\Upsilon(\xi) = \OO_{\bfs \beta}\{a \cdot \xi - a \cdot \infty, \bfs 0 \}. 
$
Then an exercise in Wick's calculus shows that
\[
\E_{\Upsilon} \Phi_{\bfs \beta}(z) = \varphi_{\bfs \beta + a \cdot \xi - a \cdot \infty}(z).
\]
As a conformal field, the latter depends on the marked points $q_k$, $\xi$, and $z$. It is a $[0,0]$-differential in all of the $q$-variables and a PPS$(ib, -ib)$ form in the $z$ variable. By Theorem \ref{Thm_SLE MO}, it is a martingale observable for the SLE process with driving function
\[
d \xi_t = \sqrt{\kappa} \, d B_t + \Lambda_{D_t}(\xi_t) \, dt,
\]
with the drift function $\Lambda(\xi)$ given by
\[
\Lambda(\xi) := \kappa \pa_{\xi} \log \E \Upsilon(\xi) - r_0(\xi).
\]
Here $D_t = g_t(D_0 \backslash \gamma[0,t])$, where $\gamma$ is the SLE curve. Thus $D_t$ is yet another chordal standard domain, and the field $\Lambda_{D_t}$ is evaluated in this domain. The nature of the OPE exponential $\Upsilon(\xi)$ makes it difficult to give an exact expression for the term $\E \Upsilon(\xi)$, even in the standard coordinate chart. It only depends on the divisor $\bfs \beta$ and the marked point $\xi$. Similarly, the conformal field $\E \Upsilon \Phi_{\bfs \beta}(z)$ only depends on these divisors and $\xi$, along with the point $z \in D_0$. The coordinate chart $g_t^{-1}$ is well-defined around these marked points, so by evaluating the field $\E_{\Upsilon} \Phi_{\bfs \beta}(z)$ along the charts $g_t^{-1}$ the martingale in coordinates is
\[
t \mapsto -2b \arg g_t'(z) + 2 \sum_k \beta_k \Im \, G_D^+(g_t(q_k), g_t(z))
\]
by \eqref{Phi beta}. For $z$ on the real line (but distinct from the $q_k$ and $\xi$) it can be shown that $\E_{\Upsilon} \Phi_{\bfs \beta}(z)$ is piecewise constant.
\end{example}

In the multiply connected setting it is very difficult to provide any explicit formulas for the martingales, other than in terms of complicated special functions. Nonetheless we can see interesting functional relations arise between the correlation functions, even without explicit expressions available for them.

\begin{example}{A functional relation for a sequence of correlation functions derived from $T_{\bfs \beta}$}
In the introduction we argued that CFT can be used to derive functional equations between correlation functions. Via the Virasoro field $T_{\bfs \beta}$ we can use the Eguchi-Ooguri form of Ward's equation to derive a recursive sequence of correlation functions. For the setup, let $D$ be a chordal standard domain and fix a sequence of distinct points $z_1, z_2, \ldots \in D \cup \R$. Fix a background charge $\bfs \beta = \sum_k \beta_k \cdot q_k$ where one has $q_k \in \R$ and the $\beta_k$ satisfy the neutrality condition $\sum \beta_k = b \chi$. Assume that the $q_k$ are all distinct from the $z_j$. We will consider the conformal fields 
\[
V_n = \E \Upsilon(\xi) \XX_n, \qquad \XX_n = T_{\bfs \beta}(z_1) \ldots T_{\bfs \beta}(z_n),
\]
where in this case $\Upsilon(\xi) = \OO_{\bfs \beta}\{a \cdot \xi - a \cdot \infty, \bfs 0 \}$ and we recall that $a$ is a solution to $2a(a+b) = 1$. Here $\xi$ is distinct from the $q_k$ and $z_j$. As a (non-random) conformal field, $V_n$ is
\begin{itemize}
\item a $[0,0]$-differential at the $q_k$, \smallskip 
\item a Schwarzian form of order $c/12 = 1/12 - b^2$ at the $z_j$, $j=1, \ldots, n$, \smallskip 
\item and an $[h_{1,2}, 0]$-differential at $\xi$, where $h_{1,2} = a^2/2 - ab$.
\end{itemize}
It will also transform as a differential at infinity, but in these calculations we only consider conformal maps that fix infinity and have the hydrodynamic normalization there, so the transformation rule at infinity does not enter in. 
More precisely, the Loewner vector field has a triple zero at infinity and the Lie derivative at infinity in Ward's equation does not appear in the identity chart of a chordal standard multiply connected domain.  
Note the sequence $\XX_{n+1}$ is built up by successive multiplications by the stress tensor $T_{\bfs \beta}$, i.e. there is the recursive relation $\XX_{n+1} = T_{\bfs \beta}(z_{n+1}) \XX_{n}$. At the same time, the Ward equation relates multiplication by the stress tensor to the Lie derivative. By the Ward equation \eqref{eq: Ward_no_integral}, and recalling that $A_{\bfs \beta} = T_{\bfs \beta} - \E\,T_{\bfs \beta}$, we obtain the recursive equation
\begin{align}\label{eq: example_V_recursion} 
2 V_{n+1} = 2 V_n \E\,T_{\bfs \beta}(z_{n+1}) + (\LL_{v_{z_{n+1} }}^+ + \nabla_{ \mathbb{H}_g, \bfs q }) V_n.
\end{align}
In a simply connected domain, this recursion formula coincides with that for an avoiding probability of the SLE$(8/3)$ curve (in the $c=0$ case), see \cite[Section 14.5]{KM13} and \cite{FW03}. This is closely related to the restriction property and we refer to \cite{AB24} for a related recent work in a multiply connected domain. 
In principle, \eqref{eq: Virasoro_recursion} can be solved given the one-point function $\E\,T_{\bfs \beta}(z)$ and an initial value $V_0 = \E \Upsilon(\xi) $. 
Although $V_n$ is not a martingale observable, Theorem \ref{Thm_SLE MO} states that $M_n := V_n/\E \Upsilon(\xi)$ is a martingale observable for an SLE process, when evaluated on the random charts $g_t^{-1}$. By dividing through \eqref{eq: example_V_recursion} by $\E \Upsilon(\xi)$ and using the Leibniz rule we obtain
\begin{align}
\begin{split} \label{eq: example_M_recursion} 
2 M_{n+1} &= 2 M_n \E \, T_{\bfs \beta}(z_{n+1}) + (\LL_{v_{z_{n+1} }}^+ + \nabla_{\mathbb{H}_{g}, \bfs q }) M_n + \frac{M_n}{\E \Upsilon(\xi)} (\LL_{v_{z_{n+1} }}^+  + \nabla_{\mathbb{H}_{g},\bfs q}) \E \Upsilon(\xi). 
\end{split}
\end{align} 
In the higher genus case that we consider both will be expressed in terms of complicated special functions, which makes it difficult to find a closed form solution to the recursive equation. The transformation rules for $V_n$ allow us to compute the Lie derivative term as 
\begin{align*}
\LL_{v_{z_{n+1}}}^+ V_n & = v_{ z_{n+1} }(\xi) \pa_{\xi} V_n + h_{1,2} v_{ z_{n+1} }'(\xi) V_n 
+ \sum_{j=1}^n \left[ v_{ z_{n+1} }(z_j) \pa_{z_j} + 2 v_{ z_{n+1} }'(z_j) \right] V_n + \frac{c}{12} \sum_{j=1}^n v_{ z_{n+1} }'''(z_j) V_{n,j},
\end{align*}
where $V_{n,j} = \E \Upsilon(\xi) \XX_{n,j}$ and $\XX_{n,j} = T_{\bfs \beta}(z_1) \ldots \hat{T}_{\bfs \beta}(z_j) \ldots T_{\bfs \beta}(z_n)$. Here, the hat indicates that term is left out of the product. 
Combining the above, it follows that 
\begin{align} \label{eq: Virasoro_recursion}
\begin{split}
2 V_{n+1} & = \bigg[  2  \E\,T_{\bfs \beta}(z_{n+1})  +   \sum_j \Big( v_{ z_{n+1} }(z_j) \pa_{z_j} + 2 v_{ z_{n+1} }'(z_j) \Big)  + \mathfrak{D}_{n}[\xi,\bfs q]  \bigg] V_n 
+ \frac{c}{12} \sum_j v_{ z_{n+1} }'''(z_j) V_{n,j},  
\end{split}
\end{align}
where 
\begin{align*}
\mathfrak{D}_{n}[\xi,\bfs q] &:=    h_{1,2} v_{ z_{n+1} }'(\xi) +  v_{ z_{n+1} }(\xi) \pa_{\xi}    + \nabla_{ \mathbb{H}_g, \bfs q }. 
\end{align*} 
After normalizing \eqref{eq: Virasoro_recursion} the sequence $M_n$ satisfies 
\begin{align}
\begin{split}
2 M_{n+1} &  = \bigg[  2  \E\,T_{\bfs \beta}(z_{n+1})   +   \sum_j \Big( v_{ z_{n+1} }(z_j) \pa_{z_j} + 2 v_{ z_{n+1} }'(z_j) \Big)   + \mathfrak{D}_{n}[\xi,\bfs q]  \bigg] M_n \\
&\quad + \frac{c}{12} \sum_j v_{ z_{n+1} }'''(z_j) M_{n,j} +  \frac{\mathfrak{D}_{n}[\xi,\bfs q] \E \Upsilon(\xi) }{ \E \Upsilon(\xi) }  M_n, 
\end{split}
\end{align}
where $M_{n,j} = V_{n,j} / \E \Upsilon(\xi)$. This formula agrees with \eqref{eq: example_M_recursion} since each $M_n$ is a $[0,0]$-differential at the $q_k$ and at $\xi$, and a Schwarzian form of order $c/12$ at $z_1, \ldots, z_n$, whereas $\E \Upsilon(\xi)$ is an $[h_{1,2},0]$-differential at $\xi.$    
\end{example}

\end{document}